\documentclass{LMCS}

\usepackage{latexsym}
\usepackage{xspace}
\usepackage{eepic}
\usepackage{times}
\usepackage{amsmath}
\usepackage{graphics}
\usepackage{color}
\usepackage{epsfig}
\usepackage{enumerate,hyperref}

\def\exp{{\sc Exptime}}
\def\nlsp{{\sc NLogSpace}}
\def\aa{\alpha}
\def\a-{\aa^-}
\newcommand{\Nat}{\mbox{I$\!$N}}

\def\pa-{p_{\alpha^{-}}}

\def\head{\mathsf{head}}
\newcommand{\FF}{{\mathcal F}}
\newcommand{\val}{{\mathcal V}}
\def\genA{\pairs{\Sigma, b, Q, \delta, q_0, \FF}}
\newcommand{\strat}{\ensuremath{\mathsf{str}}\xspace}
\newcommand{\prom}{\ensuremath{\mathsf{pro}}\xspace}
\newcommand{\ann}{\ensuremath{\mathsf{ann}}\xspace}
\newcommand{\safety}{\textsc{Safety }}
\newcommand{\forrall}{\textsc{Forall }}
\newcommand{\pairs}[1]{\langle #1 \rangle}
\newcommand{\ppairs}[1]{\langle\!\langle #1 \rangle\!\rangle}

\def\doi{4 (3:11) 2008}
\lmcsheading%
{\doi}
{1--27}
{}
{}
{Jan.~\phantom{0}7, 2007}
{Sep.~22, 2008}
{}   

\begin{document}

\title[The Complexity of Enriched $\mu$-Calculi]{The Complexity of Enriched \texorpdfstring{$\mu$}{mu}-Calculi\rsuper*}

\titlecomment{{\lsuper*}A preliminary version of this paper appears in the Proceedings of the 33rd International Colloquium on Automata, Languages and Programming, 2006.}

\author[P.~A.~Bonatti]{Piero A. Bonatti\rsuper a} 
\address{{\lsuper a}Universit\`{a} di Napoli ``Federico II'',
Dipartimento di Scienze Fisiche, 80126 Napoli, Italy}
\email{bonatti@na.infn.it}
\thanks{{\lsuper a}Supported in part by the European Network of
Excellence REWERSE, IST-2004-506779.}

\author[C.~Lutz]{Carsten Lutz\rsuper b} 
\address{{\lsuper b}TU Dresden, Institute for Theoretical Computer Science,
01062 Dresden, Germany}
\email{clu@tcs.inf.tu-dresden.de}

\author[A.~Murano]{Aniello Murano\rsuper c} 
\address{{\lsuper c}Universit\`{a} di Napoli ``Federico II'',
Dipartimento di Scienze Fisiche, 80126 Napoli, Italy}
\email{murano@na.infn.it} 

\author[M.~Y.~Vardi]{Moshe Y. Vardi\rsuper d} 
\address{{\lsuper d} {Microsoft Research} and {Rice University, 
Dept. of Computer Science, TX 77251-1892, USA}}
\email{vardi@cs.rice.edu}
\thanks{{\lsuper d} Supported in part by NSF grants CCR-0311326 and
ANI-0216467, by BSF grant 9800096, and by Texas ATP grant 003604-0058-2003.
Work done in part while this author was visiting the Isaac Newton Institute for
Mathematical Science, Cambridge, UK, as part of a Special Programme on Logic
and Algorithm.}

\keywords{$\mu$-calculi, expressive description logics, hybrid modal
  logics, fully enriched automata, 2-way graded alternating parity
  automata} 
\subjclass{F.3.1, F.4.1}



\begin{abstract}
\noindent
  The \emph{fully enriched $\mu$-calculus} is the extension of the
  propositional $\mu$-calculus with inverse programs, graded
  modalities, and nominals. While satisfiability in several expressive
  fragments of the fully enriched $\mu$-calculus is known to be
  decidable and {\sc ExpTime}-complete, it has recently been proved
  that the full calculus is undecidable. In this paper, we study the
  fragments of the fully enriched $\mu$-calculus that are obtained by
  dropping at least one of the additional constructs.  We show that,
  in all fragments obtained in this way, satisfiability is decidable
  and {\sc ExpTime}-complete. Thus, we identify a family of decidable
  logics that are maximal (and incomparable) in expressive power. Our
  results are obtained by introducing two new automata models, showing
  that their emptiness problems are {\sc ExpTime}-complete, and then
  reducing satisfiability in the relevant logics to these problems.  The
  automata models we introduce are \emph{two-way graded alternating
    parity automata} over infinite trees (2GAPTs) and \emph{fully
    enriched automata} (FEAs) over infinite forests.  The former are a
  common generalization of two incomparable automata models from the
  literature. The latter extend alternating automata in a similar way
  as the fully enriched $\mu$-calculus extends the standard
  $\mu$-calculus.
\end{abstract}

\maketitle
\vfill\eject

\section{Introduction}

The \emph{$\mu$-calculus} is a propositional modal logic augmented
with least and greatest fixpoint operators \cite{Koz83}. It is often
used as a target formalism for embedding temporal and modal logics
with the goal of transferring computational and model-theoretic
properties such as the {\sc ExpTime} upper complexity bound.
\emph{Description logics (DLs)} are a family of knowledge
representation languages that originated in artificial
intelligence~\cite{Baader-et-al-01} and currently receive considerable
attention, which is mainly due to their use as an ontology language in
prominent applications such as the semantic web
\cite{Baader-et-al-02b}. Notably, DLs have recently been standardized
as the ontology language OWL by the W3C committee. It has been pointed
out by several authors that, by embedding DLs into the $\mu$-calculus,
we can identify DLs that are of very high expressive power, but
computationally well-behaved \cite{CGL01,SV01,KSV02}. When putting
this idea to work, we face the problem that modern DLs such as the
ones underlying OWL include several constructs that cannot easily be
translated into the $\mu$-calculus. The most important such constructs
are inverse programs, graded modalities, and nominals. Intuitively,
inverse programs allow to travel backwards along accessibility
relations \cite{Var98}, nominals are propositional variables
interpreted as singleton sets \cite{SV01}, and graded modalities
enable statements about the number of successors (and possibly
predecessors) of a state \cite{KSV02}. All of the mentioned constructs
are available in the DLs underlying OWL.

The extension of the $\mu$-calculus with these constructs induces a
family of enriched $\mu$-calculi. These calculi may or may not enjoy
the attractive computational properties of the original
$\mu$-calculus: on the one hand, it has been shown that
satisfiability in a number of the enriched calculi is decidable and
{\sc ExpTime}-complete \cite{CGL01,SV01,KSV02}. On the other hand,
it has recently been proved by Bonatti and Peron that satisfiability
is undecidable in the \emph{fully enriched $\mu$-calculus}, i.e.,
the logic obtained by extending the $\mu$-calculus with all of the
above constructs simultaneously \cite{BP04}. In computer science
logic, it has always been a major research goal to identify
decidable logics that are as expressive as possible. Thus, the above
results raise the question of maximal decidable fragments of the
fully enriched $\mu$-calculus.  In this paper, we study this
question in a systematic way by considering all fragments of the
fully enriched $\mu$-calculus that are obtained by dropping at least
one of inverse programs, graded modalities, and nominals. We show
that, in all these fragments, satisfiability is decidable and {\sc
ExpTime}-complete. Thus, we identify a whole family of decidable
logics that have maximum  expressivity.

The relevant fragments of the fully enriched $\mu$-calculus are shown in
Figure~\ref{fig:fragments} together with the complexity of their satisfiability
problem. The results shown in gray are already known from the literature: the
{\sc ExpTime} lower bound for the original $\mu$-calculus stems from
\cite{FL79};
\definecolor{gray}{rgb}{0.4,0.4,0.4}
\begin{figure}[tb]
  \centering
  \begin{tabular}{|l|c|c|c|c|}
  \hline
   & Inverse progr. & Graded mod. & Nominals & \multicolumn{1}{c|}{Complexity}\\
  \hline
  fully enriched  $\mu$-calculus & x & x & x & {\color{gray} undecidable} \\
  full graded  $\mu$-calculus & x & x & &{\sc ExpTime} ({\color{gray} 1ary}/2ary)\\
  full hybrid  $\mu$-calculus & x & & x & {\sc\color{gray} ExpTime} \\
  hybrid graded  $\mu$-calculus &  & x & x & {\sc ExpTime} (1ary/2ary)\\
   graded  $\mu$-calculus & & x & & {\color{gray} {\sc ExpTime} (1ary/2ary)}\\
  \hline
  \end{tabular}
  \caption{Enriched $\mu$-calculi and previous results.}
  \label{fig:fragments}
\end{figure}
it has been shown in \cite{SV01} that satisfiability in the full
hybrid $\mu$-calculus is in {\sc ExpTime}; under the assumption that
the numbers inside graded modalities are coded in unary, the same
result was proved for the full graded $\mu$-calculus in \cite{CGL01};
finally, the same was also shown for the (non-full) graded
$\mu$-calculus in \cite{KSV02} under the assumption of binary coding.
In this paper, we prove {\sc ExpTime}-completeness of the full graded
$\mu$-calculus and the hybrid graded $\mu$-calculus.  In both cases,
we allow numbers to be coded in binary (in contrast, the techniques
used in \cite{CGL01} involve an exponential blow-up when numbers are
coded in binary).

Our results are based on the automata-theoretic approach and extends
the techniques in \cite{KSV02,SV01,Var98}. They involve introducing
two novel automata models. To show that the full graded $\mu$-calculus
is in {\sc ExpTime}, we introduce \emph{two-way graded parity tree
  automata (2GAPTs)}. These automata generalize in a natural way two
existing, but incomparable automata models: two-way alternating parity
tree automata (2APTs) \cite{Var98} and (one-way) graded alternating
parity tree automata (GAPTs) \cite{KSV02}. The phrase ``two-way''
indicates that 2GAPTs (like 2APTs) can move up and down in the tree.
The phrase ``graded'' indicates that 2GAPTs (like GAPTs) have the
ability to count the number of successors of a tree node that it moves
to.  Namely, such an automaton can move to at least $n$ or all but $n$
successors of the current node, without specifying which successors
exactly these are. We show that the emptines problem for 2GAPT is in
{\sc ExpTime} by a reduction to the emptiness of graded
nondeterministic parity tree automata (GNPTs) as introduced in
\cite{KSV02}. This is the technically most involved part of this
paper. To show the desired upper bound for the full graded
$\mu$-calculus, it remains to reduce satisfiability in this calculus
to emptiness of 2GAPTs. This reduction is based on the tree model
property of the full graded $\mu$-calculus, and technically rather
standard.

To show that the hybrid graded $\mu$-calculus is in {\sc ExpTime}, we
introduce \emph{fully enriched automata (FEAs)} which run on infinite
forests and, like 2GAPTs, use a parity acceptance condition.  FEAs
extend 2GAPTs by additionally allowing the automaton to send a copy of
itself to some or all roots of the forest. This feature of ``jumping
to the roots'' is in rough correspondence with the nominals included
in the full hybrid $\mu$-calculus. We show that the emptiness problem
for FEAs is in {\sc ExpTime} using an easy reduction to the emptiness
problem for 2GAPTs. To show that the hybrid graded $\mu$-calculus is
in {\sc ExpTime}, it thus remains to reduce satisfiability in this
calculus to emptiness of FEAs. Since the correspondence between
nominals in the $\mu$-calculus and the jumping to roots of FEAs is
only a rough one, this reduction is more delicate than the corresponding
one for the full graded $\mu$-calculus. The reduction is based on a
forest model property enjoyed by the hybrid graded $\mu$-calculus and
requires us to work with the \emph{two-way} automata FEAs although the
hybrid graded $\mu$-calculus does not offer inverse programs.

We remark that, intuitively, FEAs generalize alternating automata on
infinite trees in a similar way as the fully enriched $\mu$-calculus
extends the standard $\mu$-calculus: FEAs can move up to a node's
predecessor (by analogy with inverse programs), move down to at least
$n$ or all but $n$ successors (by analogy with graded modalities), and
jump directly to the roots of the input forest (which are the
analogues of nominals). Still, decidability of the emptiness problem
for FEAs does not contradict the undecidability of the fully enriched
$\mu$-calculus since the latter does not enjoy a forest model property
\cite{BP04}, and hence satisfiability cannot be decided using
forest-based FEAs.

The rest of the paper is structured as follows. The subsequent section
introduces the syntax and semantics of the fully enriched
$\mu$-calculus. The tree model property for the full graded
$\mu$-calculus and a forest model property for the hybrid graded
$\mu$-calculus are then established in Section~\ref{sect:tmpfmp}. In
Section~\ref{sec:enriched-automata}, we introduce FEAs and 2GAPTs and
show how the emptiness problem for the former can be polynomially
reduced to that of the latter. In this section, we also state our
upper bounds for the emptiness problem of these automata models. Then,
Section~\ref{sec:upper-bounds} is concerned with reducing the
satisfiability problem of enriched $\mu$-calculi to the emptiness
problems of 2GAPTs and FEAs. The purpose of
Section~\ref{sect:emptiness} is to reduce the emptiness problem for
2GAPTs to that of GNPTs. Finally, we conclude in
Section~\ref{sect:concl}.

\section{Enriched $\mu$-calculi} \label{sec:prelim}

We introduce the syntax and semantics of the fully enriched $\mu$-calculus.
Let $\mathsf{Prop}$ be a finite set of \emph{atomic propositions},
$\mathsf{Var}$ a finite set of \emph{propositional variables},
$\mathsf{Nom}$ a finite set of \emph{nominals}, and $\mathsf{Prog}$ a
finite set of \emph{atomic programs}. We use $\mathsf{Prog}^-$ to
denote the set of \emph{inverse programs} $\{ a^- \mid a \in
\mathsf{Prog} \}$. The elements of $\mathsf{Prog} \cup
\mathsf{Prog}^-$ are called \emph{programs}. We assume $a^{--}=a$. The
set of \emph{formulas of the fully enriched $\mu$-calculus} is the
smallest set such that

\begin{enumerate}[$\bullet$]
\item
$\mathsf{true}$ and $\mathsf{false}$ are formulas;
\item $p$ and $\neg p$, for $p \in \mathsf{Prop}$, are formulas;
\item $o$ and $\neg o$, for $o \in \mathsf{Nom}$, are formulas;
\item $x \in \mathsf{Var}$ is a formula;
\item $\varphi_1 \vee \varphi_2$ and $\varphi_1 \wedge \varphi_2$ are formulas if $\varphi_1$ and $\varphi_2$ are
formulas;
\item $\pairs{n,\aa} \varphi$, and $[n, \aa] \varphi$ are formulas if $n$ is a non-negative
integer, $\aa$ is a program, and $\varphi$ is a formula;
\item $\mu y. \varphi(y)$ and $\nu y. \varphi(y)$ are formulas if $y$ is a propositional
variable and $\varphi(y)$ is a formula containing $y$ as a free variable.
\end{enumerate}

Observe that we use positive normal form, i.e., negation is applied
only to atomic propositions.

We call $\mu$ and $\nu$ \emph{fixpoint operators} and use $\lambda$ to denote a
fixpoint operator $\mu$ or $\nu$. A propositional variable $y$ occurs
\emph{free} in a formula if it is not in the scope of a fixpoint operator
$\lambda y$, and \emph{bounded} otherwise. Note that $y$ may occur both bounded
and free in a formula. A \emph{sentence} is a formula that contains no free
variables. For a formula $\lambda y.  \varphi(y)$, we write $\varphi(\lambda y.
\varphi(y))$ to denote the formula that is obtained by one-step unfolding,
i.e., replacing each free occurrence of $y$ in $\varphi$ with $\lambda
y.\varphi(y)$. We often refer to the \emph{graded
  modalities} $\pairs{n, \aa} \varphi$ and $[n, \aa] \varphi$ as
\emph{atleast formulas} and \emph{allbut formulas} and assume that the integers
in these operators are given in binary coding: the contribution of $n$ to the
length of the formulas $\pairs{n,\aa}\varphi$ and $[n,\aa]\varphi$ is $\lceil
\log n \rceil$ rather than~$n$. We refer to fragments of the fully enriched
$\mu$-calculus using the names from Figure~\ref{fig:fragments}.
Hence, we say that a formula of the fully enriched $\mu$-calculus is also a
formula of the \emph{hybrid graded $\mu$-calculus}, \emph{full
  hybrid $\mu$-calculus}, and \emph{full graded $\mu$-calculus} if it
does not have inverse programs, graded modalities, and nominals,
respectively.

The semantics of the fully enriched $\mu$-calculus is defined in terms of a
\emph{Kripke structure}, i.e., a tuple $K = \pairs{W,R,L}$ where

\begin{enumerate}[$\bullet$]
   \item $W$ is a non-empty (possibly infinite) set of \emph{states};
   \item $R : \mathsf{Prog} \rightarrow 2^{W \times W}$ assigns
   to each atomic program a binary relation over
   $W$;
   \item $L:\mathsf{Prop} \cup \mathsf{Nom} \rightarrow 2^{W}$ assigns
   to each atomic proposition and nominal a set of states such that
   the sets assigned to nominals are singletons.
\end{enumerate}

To deal with inverse programs, we extend $R$ as follows: for each
atomic program $a$, we set $R(a^-) = \{(v, u) : (u, v) \in R(a)\}$.
For a program $\aa$, if $(w,w')\in R(\aa)$, we say that $w'$ is an
\emph{$\aa$-successor} of $w$. With $\mathsf{succ}_R(w, \alpha)$ we denote
the set of $\aa$-successors of $w$.

Informally, an \emph{atleast} formula $\pairs{n,\aa}\varphi$ holds at a state
$w$ of a Kripke structure $K$ if $\varphi$ holds at least in $n+1$
$\aa$-successors of $w$. Dually, the \emph{allbut} formula $[n, \aa]\varphi$
holds in a state $w$ of a Kripke structure $K$ if $\varphi$ holds in all but at
most $n$ $\aa$-successors of $w$. Note that $\neg \pairs{n,\aa}\varphi$ is
equivalent to $[n, \aa]\neg \varphi$.  Indeed,$\neg \pairs{n,\aa}\varphi$ holds
in a state $w$ if $\varphi$ holds in less than $n+1$ $\aa$-successors of $w$,
thus, at most $n$ $\aa$-successors of $w$ do not satisfy $\neg \varphi$, that
is, $[n, \aa]\neg \varphi$ holds in $w$.  The modalities $\langle \aa \rangle
\varphi$ and $[\aa] \varphi$ of the standard $\mu$-calculus can be expressed as
$\pairs{0,\aa} \varphi$ and $[0,\aa] \varphi$, respectively. The least and
greatest fixpoint operators are interpreted as in the standard $\mu$-calculus.
Readers not familiar with fixpoints might want to look at
\cite{Koz83,SE89,BS06} for instructive examples and explanations of the
semantics of the $\mu$-calculus.

To formalize the semantics, we introduce valuations. Given a Kripke structure
$K = \pairs{W,R,L}$ and a set $\{y_1,\ldots, y_n \}$ of propositional variables
in $\mathsf{Var}$, a \emph{valuation} $\val: \{y_1,\ldots,y_n \} \rightarrow
2^{W}$ is an assignment of subsets of $W$ to the variables $y_1,\ldots, y_n$.
For a valuation $\val$, a variable $y$, and a set $W' \subseteq W$, we denote
by $\val[y \leftarrow W']$ the valuation obtained from $\val$ by assigning $W'$
to $y$. A formula $\varphi$ with free variables among $y_1,\ldots,y_n$ is
interpreted over the structure $K$ as a mapping $\varphi^{K}$ from valuations
to $2^{W}$, i.e., $\varphi^{K}(\val)$ denotes the set of states that satisfy
$\varphi$ under valuation $\val$. The mapping $\varphi^{K}$ is defined
inductively as follows:

\begin{enumerate}[$\bullet$]

\item $\mathsf{true}^{K}(\val) = W$ and $\mathsf{false}^{K}(\val) = \emptyset $;

\item for $p \in \mathsf{Prop} \cup \mathsf{Nom}$,
we have $p^K(\val) = L(p)$ and $(\neg p)^{K}(\val) = W \setminus L(p)$;

\item for $y \in \mathsf{Var}$, we have $y^{K}(\val) =
\val(y)$;

\item $(\varphi_1 \wedge \varphi_2)^{K}(\val) =
\varphi_1^{K}(\val) \cap \varphi_2^{K}(\val)$

\item $(\varphi_1 \vee \varphi_2)^{K}(\val) =
\varphi_1^{K}(\val) \cup \varphi_2^{K}(\val)$;

\item $(\pairs{n,\aa} \varphi)^{K}(\val)= \{w : | \{w'\in W :
(w,w')\in R(\aa) \text{ and } w' \in \varphi^{K}(\val)\}| > n
\}$;

\item \mbox{$([n,\aa] \varphi)^{K}(\val)= \{w : | \{w'\in W :
(w,w')\in R(\aa) \text{ and } w' \not \in \varphi^{K}(\val)\}|\leq n
\}$};

\item $(\mu y.\varphi(y))^{K}(\val)=
\bigcap \{W'\subseteq W : \varphi^{K}(\val [y \leftarrow W']) \subseteq W'\}$;

\item $(\nu y.\varphi(y))^{K}(\val)=
\bigcup \{W'\subseteq W : W' \subseteq \varphi^{K}(\val [y \leftarrow W'])\}$.

\end{enumerate}

Note that, in the clauses for graded modalities, $\aa$ denotes
a program, i.e., $\aa$ can be either an atomic program or an
inverse program. Also, note that no valuation is required for a sentence.

Let $K=\langle W,R,L \rangle$ be a Kripke structure and $\varphi$ a sentence.
For a state $w \in W$, we say that $\varphi$ \emph{holds} at $w$ in $K$,
denoted $K,w \models \varphi$, if $w \in \varphi^{K}(\emptyset)$. $K$ is a
\emph{model} of $\varphi$ if there is a $w \in W$ such that $K,w \models
\varphi$. Finally, $\varphi$ is \emph{satisfiable} if it has a model.

\section{Tree and Forest Model Properties}
\label{sect:tmpfmp}

We show that the full graded $\mu$-calculus has the tree model
property, and that the hybrid graded $\mu$-calculus has a forest model
property.  Regarding the latter, we speak of ``a'' (rather than
``the'') forest model property because it is an abstraction of the
models that is forest-shaped, instead of the models themselves.

For a (potentially infinite) set $X$, we use $X^+$ ($X^*$) to denote
the set of all non-empty (possibly empty) words over $X$. As usual,
for $x,y \in X^*$, we use $x \cdot y$ to denote the concatenation of
$x$ and $y$. Also, we use $\varepsilon$ to denote the empty word and
by convention we take $x \cdot \varepsilon=x$, for each $x\in X^*$.
Let $\Nat$ be a set of non-negative integers. A \emph{forest} is a set
$F \subseteq \Nat^+$ that is prefix-closed, that is, if $x \cdot c \in
F$ with $x \in \Nat^+$ and $c \in \Nat$, then also $x \in F$. The
elements of $F$ are called \emph{nodes}. For every $x \in F$, the
nodes $x \cdot c \in F$ with $c \in \Nat$ are the \emph{successors} of
$x$, and $x$ is their \emph{predecessor}. We use $\mathsf{succ}(x)$ to
denote the set of all successors of $x$ in $F$. A \emph{leaf} is a node
without successors, and a \emph{root} is a node without predecessors.
A forest $F$ is a \emph{tree} if $F \subseteq \{ c \cdot x \mid x \in \Nat^*
\}$ for some $c \in \Nat$ (the \emph{root} of $F$). The root of a tree $F$ is
denoted with $\mathsf{root}(F)$. If for some $c$, $T = F \cap \{ c \cdot x \mid
x \in \Nat^* \}$, then we say that $T$ is \emph{the tree of} $F$ \emph{rooted
in $c$}.

We call a Kripke structure $K=\pairs{W,R,L}$ a \emph{forest
structure} if
\begin{enumerate}[(i)]

\item $W$ is a forest,

\item $\bigcup_{\alpha \in \mathsf{Prog} \cup \mathsf{Prog}^-} R(\alpha)
= \{ (w,v) \in W \times W \mid \text{$w$ is a predecessor or a
  successor of $v$} \}$.

\end{enumerate}
Moreover, $K$ is \emph{directed} if $(w,v)\in\bigcup_{a \in
\mathsf{Prog}} R(a)$ implies that $v$ is a successor of $w$. If $W$
is a tree, then we call $K$ a \emph{tree structure}.

We call $K=\langle W,R,L \rangle$ a \emph{directed quasi-forest
  structure} if $\langle W,R',L \rangle$ is a directed forest
structure, where $R'(a) = R(a) \setminus (W \times \Nat)$ for all
$a\in \mathsf{Prog}$, i.e., $K$ becomes a directed forest structure
after deleting all the edges entering a root of $W$.
Let $\varphi$ be a formula and $o_1,\dots,o_k$ the nominals
occurring in $\varphi$.  A \emph{forest model} (resp.\ \emph{tree
model}, \emph{quasi-forest model}) of $\varphi$ is a forest (resp.\
tree, quasi-forest) structure $K=\langle W,R,L \rangle$ such that
there are roots $c_0,\dots,c_k \in W \cap \Nat$ with $K,c_0 \models
\varphi$ and $L(o_i)=\{c_i\}$, for $1 \leq i \leq k$.  Observe that
the roots $c_0,\dots,c_k$ do not have to be distinct.

Using a standard unwinding technique such as in \cite{Var98,KSV02}, it
is possible to show that the full graded $\mu$-calculus enjoys the
tree model property, i.e., if a formula $\varphi$ is satisfiable, it
is also satisfiable in a tree model. We omit details and concentrate
on the similar, but more difficult proof of the fact that the hybrid
graded $\mu$-calculus has a forest model property.
\begin{thm}
  \label{tree1} If a sentence $\varphi$ of the full graded
  $\mu$-calculus is satisfiable, then $\varphi$ has a tree
  model.
\end{thm}
In contrast to the full graded $\mu$-calculus, the hybrid graded $\mu$-calculus
does not enjoy the tree model property. This is, for example, witnessed by the
formula
$$
o \wedge \pairs{0, a} (p_1 \wedge \pairs{0, a} (p_2 \wedge \cdots
\pairs{0, a} (p_{n-1} \wedge \pairs{0, a} o) \cdots ))
$$
which generates a cycle of length $n$ if the atomic propositions $p_i$
are forced to be mutually exclusive (which is easy using additional
formulas). However, we can follow \cite{SV01,KSV02} to show that the
hybrid graded $\mu$-calculus has a forest model property. More
precisely, we prove that the hybrid graded $\mu$-calculus enjoys the
\emph{quasi-forest model property}, i.e., if a formula $\varphi$ is
satisfiable, it is also satisfiable in a directed quasi-forest
structure.

The proof is a variation of the original construction for the $\mu$-calculus
given by Streett and Emerson in \cite{SE89}. It is an amalgamation of the
constructions for the hybrid $\mu$-calculus in \cite{SV01} and for the hybrid
graded $\mu$-calculus in \cite{KSV02}. We start with introducing the notion of
a well-founded adorned pre-model, which augments a model with additional
information that is relevant for the evaluation of fixpoint formulas.  Then, we
show that any satisfiable sentence $\varphi$ of the hybrid graded
$\mu$-calculus has a well-founded adorned pre-model, and that any such
pre-model can be unwound into a tree-shaped one, which can be converted into a
directed quasi-forest model of $\varphi$.

To determine the truth value of a Boolean formula, it suffices to
consider its subformulas. For $\mu$-calculus formulas, one has to
consider a larger collection of formulas, the so called Fischer-Ladner
closure \cite{FL79}.
The \emph{closure} $\mathsf{cl}(\varphi)$ of a sentence $\varphi$ of the hybrid
graded $\mu$-calculus is the smallest set of sentences satisfying the
following:
\begin{enumerate}[$\bullet$]

\item $\varphi \in \mathsf{cl}(\varphi)$;

\item if $\psi_1 \wedge \psi_2 \in \mathsf{cl}(\varphi)$ or $\psi_1 \vee
  \psi_2 \in \mathsf{cl}(\varphi)$, then $\{ \psi_1, \psi_2 \} \subseteq
  \mathsf{cl}(\varphi)$;

\item if $\pairs{n,a}\psi \in \mathsf{cl}(\varphi)$ or $[n,a]\psi
  \in \mathsf{cl}(\varphi)$, then $\psi \in \mathsf{cl}(\varphi)$;

\item if $\lambda y. \psi(y) \in \mathsf{cl}(\varphi)$, then $
  \psi(\lambda y. \psi(y)) \in \mathsf{cl}(\varphi)$.
\end{enumerate}
An \emph{atom} is a subset $A\subseteq \mathsf{cl}(\varphi)$ satisfying the
following properties:

\begin{enumerate}[$\bullet$]
  \item if $p \in \mathsf{Prop} \cup \mathsf{Nom}$ occurs in $\varphi$, then $p \in A$ iff $\neg p \not \in A$;
  \item if $\psi_1 \wedge \psi_2 \in \mathsf{cl}(\varphi)$, then $\psi_1 \wedge \psi_2 \in A$ iff
  $\{\psi_1,\psi_2\}\subseteq A$;
  \item if $\psi_1 \vee \psi_2 \in \mathsf{cl}(\varphi)$, then $\psi_1 \vee \psi_2 \in A$ iff
  $\{\psi_1,\psi_2\}\cap A \neq \emptyset$;
  \item if $\lambda y. \psi(y) \in \mathsf{cl}(\varphi)$, then $\lambda y. \psi(y) \in A$ iff
  $\psi(\lambda y. \psi(y))\in A$.

\end{enumerate}
The set of atoms of $\varphi$ is denoted $\mathsf{at}(\varphi)$.
A \emph{pre-model} $\pairs{K,\pi}$ for a sentence $\varphi$ of the hybrid
graded $\mu$-calculus consists of a Kripke structure $K =\pairs{W,R,L}$ and a
mapping $\pi : W \rightarrow \mathsf{at}(\varphi)$ that satisfies the following
properties:

\begin{enumerate}[$\bullet$]
    \item there is $w_0 \in W$ with $\varphi \in \pi(w_0)$;
    \item for $p \in \mathsf{Prop} \cup \mathsf{Nom}$, if $p \in \pi(w)$, then $w \in L(p)$,
            and if $\neg p \in \pi(w)$, then $w \not \in L(p)$;
    \item if $\pairs{n,a}\psi \in \pi(w)$, then there is a set $V \subseteq \mathsf{succ}_R(w, a)$,
            such that $|V| > n$ and $\psi \in \pi(v)$ for all $v \in V$;
    \item if $[n,a]\psi \in \pi(w)$, then there is a set $V \subseteq \mathsf{succ}_R(w, a)$,
            such that $|V| \leq n$ and $\psi \in \pi(v)$ for all $v \in \mathsf{succ}_R(w, a) \setminus V$.

\end{enumerate}
If there is a pre-model $\langle K,\pi \rangle$ of $\varphi$ such that
for every state $w$ and all $\psi \in \pi(w)$, it holds that $K,w
\models \psi$, then $K$ is clearly a model of $\varphi$. However, the
definition of pre-models does not guarantee that
$\psi \in \pi(w)$ is satisfied at $w$ if $\psi$ is a least fixpoint
formula.  In a nutshell, the standard approach for dealing with this
problem is to enforce that it is possible to trace the evaluation of
a least fixpoint formula through $K$ such that the original formula is
not regenerated infinitely often. When tracing such evaluations, a
complication is introduced by disjunctions and at least restrictions,
which require us to make a choice on how to continue the trace.  To
address this issue, we adapt the notion of a choice function of
Streett and Emerson \cite{SE89} to the hybrid graded $\mu$-calculus.

A \emph{choice function} for a pre-model $\pairs{K,\pi}$ for
$\varphi$ is a partial function $\mathsf{ch}$ from $W \times
\mathsf{cl}(\varphi)$ to $\mathsf{cl}(\varphi) \cup 2^W$, such that
for all $w \in W$, the following conditions hold:

\begin{enumerate}[$\bullet$]

    \item if $\psi_1 \vee \psi_2 \in \pi(w)$,
            then $\mathsf{ch}(w, \psi_1 \vee \psi_2) \in \{\psi_1, \psi_2\} \cap \pi(w)$;
    \item if $\pairs{n,a}\psi \in \pi(w)$,
            then $\mathsf{ch}(w, \pairs{n,a}\psi) = V \subseteq \mathsf{succ}_R(w, a)$,
            such that $|V| > n$ and $\psi \in \pi(v)$ for all $v \in V$;
    \item if $[n,a]\psi \in \pi(w)$,
            then $\mathsf{ch}(w, [n,a]\psi) = V \subseteq \mathsf{succ}_R(w, a)$,
            such that $|V| \leq n$ and $\psi \in \pi(v)$ for all $v \in \mathsf{succ}_R(w, a) \setminus V$.
\end{enumerate}
An \emph{adorned pre-model} $\pairs{K,\pi, \mathsf{ch}}$ of $\varphi$
consists of a pre-model $\pairs{K,\pi}$ of $\varphi$ and a choice
function~$\mathsf{ch}$. We now define the notion of a derivation
between occurrences of sentences in adorned pre-models, which
formalizes the tracing mentioned above.  For an adorned pre-model
$\pairs{K,\pi, \mathsf{ch}}$ of $\varphi$, the \emph{derivation
  relation} ${\rightsquigarrow} \subseteq (W \times
\mathsf{cl}(\varphi)) \times (W \times \mathsf{cl}(\varphi))$ is
the smallest relation such that, for all $w \in W$, we have:
\begin{enumerate}[$\bullet$]
    \item if $\psi_1 \vee \psi_2 \in \pi(w)$,
            then $(w, \psi_1 \vee \psi_2) \rightsquigarrow (w, \mathsf{ch}(\psi_1 \vee \psi_2))$;
    \item if $\psi_1 \wedge  \psi_2 \in \pi(w)$,
            then $(w, \psi_1 \wedge \psi_2) \rightsquigarrow (w, \psi_1)$ and
            $(w, \psi_1 \wedge \psi_2) \rightsquigarrow (w, \psi_2) $;
    \item if $\pairs{n,a}\psi \in \pi(w)$,
            then $(w, \pairs{n,a}\psi)\rightsquigarrow (v, \psi)$
            for each $v \in  \mathsf{ch}(w, \pairs{n,a}\psi)$;
    \item if $[n,a]\psi \in \pi(w)$,
            then $(w, [n,a]\psi)\rightsquigarrow (v, \psi)$
            for each $v \in  \mathsf{succ}_R(w, a) \setminus \mathsf{ch}(w, [n,a]\psi)$;
    \item if $\lambda y. \psi(y) \in \pi(w)$,
            then $(w, \lambda y. \psi(y)) \rightsquigarrow (w, \psi(\lambda y. \psi(y)))$.
\end{enumerate}
A least fixpoint sentence $\mu y. \psi(y)$ is \emph{regenerated} from state $w$
to state $v$ in an adorned pre-model $\pairs{K, \pi,
  \mathsf{ch}}$ of $\varphi$ if there is a sequence $(w_1, \rho_1),
\ldots, (w_k, \rho_k) \in (W \times \mathsf{cl}(\varphi))^*$, $k >1$,
such that $\rho_1=\rho_k =\mu y. \psi(y)$, $w=w_1$, $v=w_k$, the
formula $\mu y. \psi(y)$ is a sub-sentence of each $\rho_i$ in the
sequence, and for all $1 \leq i < k$, we have $(w_i, \rho_i)
\rightsquigarrow (w_{i+1}, \rho_{i+1})$. We say that $\pairs{K,\pi,
  \mathsf{ch}}$ is \emph{well-founded} if there is no least fixpoint
sentence $\mu y. \psi(y) \in \mathsf{cl}(\varphi)$ and infinite
sequence $w_1, w_2, \ldots$ such that, for each $i \geq 1$, $\mu y.
\psi(y)$ is regenerated from $w_i$ to $w_{i+1}$.  The proof of the
following lemma is based on signatures, i.e., sequence of ordinals
that guides the evaluation of least fixpoints. It is a minor variation
of the one given for the original $\mu$-calculus in \cite{SE89}. Details
are omitted.
\begin{lem}
\label{lem:modtopre}
Let $\varphi$ be a sentence of the hybrid graded $\mu$-calculus. Then:
\begin{enumerate}[\em(1)]
\item if $\varphi$ is satisfiable, it has a well-founded adorned
  pre-model;
\item
if $\pairs{K,\pi,\mathsf{ch}}$ is a well-founded adorned pre-model of
$\varphi$, then $K$ is a model of $\varphi$.
\end{enumerate}
\end{lem}
We now establish the forest model property of the hybrid graded $\mu$-calculus.
\begin{thm}
  \label{tree2}
  If a sentence $\varphi$ of the hybrid graded $\mu$-calculus is
  satisfiable, then $\varphi$ has a directed quasi-forest model.
\end{thm}
\begin{proof}
  Let $\varphi$ be satisfiable. By item~(1) of
  Lemma~\ref{lem:modtopre}, there is a well-founded adorned pre-model
  $\pairs{K,\pi, \mathsf{ch}}$ for $\varphi$.
  We unwind $K$ into a directed quasi-forest structure
  $K'=\pairs{W',L',R'}$, and define a corresponding mapping $\pi':W'
  \rightarrow \mathsf{at}(\varphi)$ and choice function $\mathsf{ch}'$
  such that $\pairs{K',\pi', \mathsf{ch}'}$ is again a well-founded
  adorned pre-model of $\varphi$. Then, item~(2) of
  Lemma~\ref{lem:modtopre} yields that $K'$ is actually a model of
  $\varphi$.

  Let $K=\pairs{W,L,R}$, and let $w_0 \in W$ such that $\varphi \in
  \pi(w_0)$. The set of states $W'$ of $K'$ is a subset of $\Nat^+$
  as required by the definition of (quasi) forest structures, and we
  define $K'$ in a stepwise manner by proceeding inductively on the
  length of elements of $W'$. Simultaneously, we define $\pi'$,
  $\mathsf{ch}'$, and a mapping $\tau:W' \rightarrow W$ that keeps
  track of correspondences between states in $K'$ and $K$.

  The base of the induction is as follows. Let $I = \{ w_1,\dots,w_k
  \} \subseteq W$ be a minimal subset such that $w_0 \in I$ and if $o$
  is a nominal in $\varphi$ and $L(o)=\{w\}$, then $w \in I$. Define
  $K'$ by setting:
  \begin{enumerate}[$\bullet$]

  \item $W' := \{ 1,\dots,k\}$;

  \item $R'(a):=\{(i,j) \mid (w_i,w_j) \in R(a), 1 \leq i \leq j \leq
    k \}$ for all $a \in \mathsf{Prog}$;

  \item $L'(p):=\{ i \mid w_i \in L(p), 1 \leq i \leq k \}$
    for all $p \in \mathsf{Prop} \cup \mathsf{Nom}$.
  \end{enumerate}
  Define $\tau$ by setting $\tau(i)=w_i$ for $1 \leq i \leq k$.  Then,
  $\pi'(w)$ is defined as $\pi(\tau(w))$ for all $w \in W'$, and
  $\mathsf{ch}'$ is defined by setting $\mathsf{ch'}(w, \psi_1 \vee
  \psi_2) = \mathsf{ch}(\tau(w), \psi_1 \vee \psi_2)$ for all $\psi_1
  \vee \psi_2 \in \pi'(w)$. Choices for atleast and allbut formulas
  are defined in the induction step.

  In the induction step, we iterate over all $w \in W'$ of maximal
  length, and for each such $w$ extend $K'$, $\pi'$, $\mathsf{ch}'$,
  and $\tau$ as follows. Let $(\langle a_1, n_1 \rangle \psi_1,v_1),
  \dots, (\langle a_m, n_m \rangle \psi_m,v_m)$ be all pairs from
  $\mathsf{cl}(\varphi) \times W$ of this form such that for each
  $(\langle a_i, n_i \rangle \psi_i,v_i)$, we have $\langle a_i, n_i
  \rangle \psi_i \in \pi(w)$ and $v_i \in \mathsf{ch}(\tau(w),\langle
  a_i, n_i \rangle \psi_i)$.  For $1 \leq i \leq m$, define
  $$
    \sigma(v_i) = \left \{
      \begin{array}{ll}
        j & \text{ if }  v_i = \tau(j), 1 \leq j \leq k \\
        w \cdot i & \text{otherwise}.
      \end{array}
    \right .
  $$
  To extend $K'$, set
  \begin{enumerate}[$\bullet$]

  \item $W' := W' \cup \{ \sigma(v_1),\dots,\sigma(v_m) \}$;

  \item $R'(a):= R'(a) \cup \{ (w,\sigma(v_i)) \mid a_i = a, 1 \leq i
    \leq m \}$ for all $a \in \mathsf{Prog}$;

  \item $L'(p):= L'(p) \cup \{ w \cdot i \in W \mid v_i \in
    L(p), 1 \leq i \leq m \}$ for all $p \in \mathsf{Prop} \cup
    \mathsf{Nom}$.

  \end{enumerate}
  Extend $\tau$ and $\pi'$ by setting $\tau(w \cdot i)=v_i$ and
  $\pi'(w \cdot i)=\pi(v_i)$ for all $w \cdot i \in W'$.  Finally,
  extend $\mathsf{ch'}$ by setting
  \begin{enumerate}[$\bullet$]

  \item $\mathsf{ch}'(w \cdot i, \psi_1 \vee \psi_2) :=
    \mathsf{ch}(v_i, \psi_1 \vee \psi_2)$ for all $w \cdot i \in W'$
    and $\psi_1 \vee \psi_2 \in \pi'(w \cdot i)$;

  \item $\mathsf{ch'}(w, \langle n,a \rangle \psi):= \{\sigma(v)
    \mid v \in \mathsf{ch}(\tau(w), \langle n,a \rangle\psi)\}$ for
    all $\langle n,a \rangle \psi \in \pi'(w)$;

  \item $\mathsf{ch'}(w, [ n,a ] \psi):= \{\sigma(v)
    \mid v \in \mathsf{ch}(\tau(w), [ n,a ] \psi)
    \cap \{ v_1,\dots,v_m\} \}$ for
    all $[ n,a ] \psi \in \pi'(w)$.

  \end{enumerate}
  It is easily seen that $K'$ is a directed quasi-forest structure.
  Since $\pairs{K,\pi,\mathsf{ch}}$ is an adorned pre-model of
  $\varphi$, it is readily checked that $\pairs{K',\pi',\mathsf{ch}'}$
  is an adorned pre-model of $\varphi$ as well. If a sentence $\mu y.
  \psi(y)$ is regenerated from $x$ to $y$ in $(K', \pi',
  \mathsf{ch}')$, then $\mu y. \psi(y)$ is regenerated from $\tau(x)$
  to $\tau(y)$ in $(K, \pi, \mathsf{ch})$. It follows that
  well-foundedness of $\pairs{K,\pi,\mathsf{ch}}$ implies
  well-foundedness of $\pairs{K',\pi',\mathsf{ch}'}$.
\end{proof}
Note that the construction from this proof fails for the fully
enriched $\mu$-calculus because the unwinding of $K$ duplicates
states, and thus also duplicates incoming edges to nominals. Together
with inverse programs and graded modalities, this may result in
$\pairs{K',\pi'}$ not being a pre-model of~$\varphi$.

\section{Enriched automata} \label{sec:enriched-automata}

Nondeterministic automata on infinite trees are a variation of
nondeterministic automata on finite and infinite words, see
\cite{Tho90} for an introduction.  \emph{Alternating automata}, as
first introduced in \cite{MS87}, are a generalization of
nondeterministic automata.  Intuitively, while a nondeterministic
automaton that visits a node $x$ of the input tree sends one copy of
itself to each of the successors of $x$, an alternating automaton
can send several copies of itself to the same successor.  In the
two-way paradigm \cite{Var98}, an automaton can send a copy of
itself to the predecessor of $x$, too.
In graded automata \cite{KSV02}, the automaton can send copies of itself to a
number $n$ of successors, without specifying which successors these exactly
are. Our most general automata model is that of fully enriched automata, as
introduced in the next subsection. These automata work on infinite forests,
include all of the above features, and additionally have the ability to send a
copy of themselves to the roots of the forest.

\subsection{Fully enriched automata}

We start with some preliminaries.
Let $F \subseteq \Nat^+$ be a forest, $x$ a node in $F$, and $c \in
\Nat$. As a convention, we take 
$(x \cdot c) \cdot -1 = x$ and $c \cdot -1$ as undefined.
A \emph{path} $\pi$ in $F$ is a minimal set $\pi \subseteq F$ such that some
root $r$ of $F$ is contained in $\pi$ and for every $x \in \pi$, either $x$ is
a leaf or there exists a $c \in F$ such that $x \cdot c \in \pi$. Given an
alphabet $\Sigma$, a \emph{$\Sigma$-labeled
  forest} is a pair $\pairs{F, V}$, where $F$ is a forest and $V : F
\rightarrow \Sigma$ maps each node of $F$ to a letter in $\Sigma$.
We call $\pairs{F, V}$ a \emph{$\Sigma$-labeled tree} if $F$ is a tree.

For a given set $Y$, let $B^{+}(Y)$ be the set of positive Boolean
formulas over $Y$ (i.e., Boolean formulas built from elements in $Y$
using $\wedge$ and $\vee$), where we also allow the formulas
$\mathsf{true}$ and $\mathsf{false}$ and $\wedge$ has precedence over
$\vee$. For a set $X \subseteq Y$ and a formula $\theta \in B^{+}(Y)$,
we say that $X$ satisfies $\theta$ iff assigning true to elements in
$X$ and assigning false to elements in $Y \setminus X$ makes $\theta$
true.  For $b>0$, let
$$
\begin{array}{ccl}
  \ppairs{b} &=& \{ \pairs{0}, \pairs{1}, \ldots, \pairs{b}\} \\
  {[[b]]} &=& \{[0],[1],\ldots,[b]\} \\
  D_b &=& \ppairs{b} \cup [[b]]
  \cup
  \{-1, \varepsilon, \pairs{\mathsf{root}}, [\mathsf{root}]\}
\end{array}
$$
A fully enriched automaton is an automaton in which the transition
function $\delta$ maps a state $q$ and a letter $\sigma$ to a formula
in $B^{+}(D_b \times Q)$.  Intuitively, an atom $(\pairs{n}, q)$
(resp.\ $([n], q)$) means that the automaton sends copies in state $q$
to $n + 1$ (resp.\ all but $n$) different successors of the current
node, $(\varepsilon, q)$ means that the automaton sends a copy in
state $q$ to the current node, $(-1, q)$ means that the automaton
sends a copy in state $q$ to the predecessor of the current node, and
$(\pairs{\mathsf{root}},q)$ (resp.\ $([\mathsf{root}],q)$) means that
the automaton sends a copy in state $q$ to some root (resp.\ all roots). When, for
instance, the automaton is in state $q$, reads a node $x$, and
$$\delta(q , V(x)) = (-1, q_1) \wedge ((\pairs{\mathsf{root}}, q_2) \vee
([\mathsf{root}],q_3)),
$$
it sends a copy in state $q_1$ to the predecessor and either sends a
copy in state $q_2$ to some root or a copy in state $q_3$ to
all roots.

Formally, a \emph{fully enriched automaton} (FEA, for short) is a
tuple $A = \langle \Sigma$, $b$, $Q$, $\delta$, $q_0$, $\FF \rangle$,
where $\Sigma$ is a finite input alphabet, $b > 0$ is a counting
bound, $Q$ is a finite set of states, $\delta: Q \times \Sigma
\rightarrow B^{+}(D_b \times Q)$ is a transition function, $q_0 \in Q$
is an initial state, and $\FF$ is an acceptance condition. A
\emph{run} of $A$ on an input $\Sigma$-labeled forest $\pairs{F,V}$ is
an $F \times Q$-labeled tree $\pairs{T_r,r}$. Intuitively, a node in
$T_r$ labeled by $(x, q)$ describes a copy of the automaton in state
$q$ that reads the node $x$ of $F$. Runs start in the initial state at
a root and satisfy the transition relation. Thus, a run
$\pairs{T_r,r}$ has to satisfy the following conditions:
\begin{enumerate}[(i)]

\item $r(\mathsf{root}(T_r)) = (c, q_0)$ for
some root $c$ of $F$ and

\item for all $y \in T_r$ with $r(y) = (x, q)$ and
  $\delta(q , V(x)) = \theta$, there is a (possibly empty) set $S
  \subseteq D_b \times Q$ such that $S$ satisfies $\theta$ and for
  all $(d, s) \in S$, the following hold:

\begin{enumerate}[$-$]

\item If $d \in \{-1, \varepsilon \} $, then $x \cdot d$ is defined and there
  is $j \in \Nat$ such that $y \cdot j \in T_r$ and $r(y \cdot j) = (x \cdot
  d, s)$;

\item If $d = \pairs{n}$, then there is a set $M \subseteq
  \mathsf{succ}(x)$ of cardinality  $n+1$ such that for all $z
  \in M$, there is $j \in \Nat$ such that $y \cdot j
  \in T_r$and $r(y \cdot j) = (z, s)$;

\item If $d = [n]$, then there is a set $M \subseteq \mathsf{succ}(x)$
  of cardinality $n$ such that for all $z \in \mathsf{succ}(x)
  \setminus M$, there is $j \in \Nat$ such that $y \cdot j \in T_r$
  and $r(y \cdot j) = (z, s)$;

 \item If $d = \pairs{\mathsf{root}}$, then for some root $c\in F$ and some $j \in \Nat$
   such that $y \cdot j \in T_r$, it holds that $r(y \cdot j) = (c, s)$;

 \item If $d = [\mathsf{root}]$, then for each root $c\in F$ there exists $j \in \Nat$
   such that $y \cdot j \in T_r$ and $r(y \cdot j) = (c, s)$.

\end{enumerate}
\end{enumerate}
Note that if $\theta = \mathsf{true}$, then $y$ does not need to have
successors.
Moreover, since no
set $S$ satisfies $\theta = \mathsf{false}$, there cannot be any run
that takes a transition with $\theta = \mathsf{false}$.

A run $\pairs{T_r,r}$ is \emph{accepting} if all its infinite paths
satisfy the acceptance condition. We consider here the \emph{parity
  acceptance condition}~\cite{Mos84,EJ91,Tho97}, where $\FF = \{\FF_1,
\FF_2,\ldots , \FF_k \}$ is such that $\FF_1 \subseteq \FF_2 \subseteq
\ldots \subseteq \FF_k = Q$. The number $k$ of sets in $\FF$ is called
the \emph{index} of the automaton. Given a run $\pairs{T_r,r}$ and an
infinite path $\pi \subseteq T_r$, let $\mathsf{Inf}(\pi) \subseteq Q$
be the set of states $q$ such that $r(y) \in F \times \{q\}$ for
infinitely many $y \in \pi$. A path $\pi$ \emph{satisfies} a parity
acceptance condition $\FF = \{ \FF_1, \FF_2,\ldots , \FF_k \}$ if the
minimal $i$ with $\mathsf{Inf}(\pi) \cap \FF_i \neq \emptyset$ is
even. An automaton \emph{accepts} a forest iff there exists an
accepting run of the automaton on the forest. We denote by
$\mathcal{L}(A)$ the set of all $\Sigma$-labeled forests that $A$
accepts.
The \emph{emptiness problem} for FEAs is to decide, given a FEA $A$,
whether $\mathcal{L}(A)=\emptyset$.

\subsection{Two-way graded alternating parity tree automata}

A \emph{two-way graded alternating parity tree automaton (2GAPT)} is a
FEA that accepts trees (instead of forests) and cannot jump to the
root of the input tree, i.e., it does not support transitions
$\pairs{\mathsf{root}}$ and $[\mathsf{root}]$. The emptiness problem
for 2GAPTs is thus a special case of the emptiness problem for FEAs.
In the following, we give a reduction of the emptiness problem for
FEAs to the emptiness problem for 2GAPTs. This allows us to derive an
upper bound for the former problem from the upper bound for the
latter that is established in Section~\ref{sect:emptiness}.

We show how to translate a FEA $A$ into a 2GAPT $A'$ such that
$\mathcal{L}(A')$ consists of the forests accepted by $A$, encoded as
trees. The encoding that we use is straightforward: the \emph{tree
  encoding} of a $\Sigma$-labeled forest $\pairs{F,V}$ is the $\Sigma
\uplus \{ \mathsf{root} \}$-labeled tree $\pairs{T,V'}$ obtained from
$\pairs{F,V}$ by adding a fresh root labeled with $\{ \mathsf{root}
\}$ whose children are the roots of $F$.
\begin{lem}\label{tree3}
  Let $A$ be a FEA running on $\Sigma$-labeled forests with $n$
  states, index $k$ and counting bound $b$. There exists a 2GAPT
  $A'$ that
  \begin{enumerate}[\em(1)]

  \item accepts exactly the tree encodings of forests accepted by
    $A$ and

  \item  has $\mathcal{O}(n)$ states,
    index $k$, and counting bound $b$.

  \end{enumerate}
\end{lem}

\proof Suppose $A=\pairs{\Sigma,b,Q,\delta,q_0,\FF}$. Define $A'$ as
$\pairs{\Sigma \uplus \{\mathsf{root}\},b,Q', \delta', q_0', \FF'}$,
where $Q'$ and $\delta'$ are defined as follows:
$$
\begin{array}{rcl}

Q' &=& Q\uplus\{q_0',q_r\}\uplus\{\mathsf{some}_q, \mathsf{all}_q \mid q\in Q\} \\[1mm]
\delta'(q_0',\mathsf{root}) &=& (\pairs{0},q_0) \wedge ([0],q_r)\\[1mm]
\delta'(q_0',\sigma)&=&\mathsf{false} \text{ for all } \sigma \neq \{\mathsf{root}\} \\[1mm]
\delta'(q_r,\mathsf{root}) &=&\mathsf{false} \\[1mm]
\delta'(q_r,\sigma) &=& ([0],q_r) \text{ for all } \sigma \neq \{\mathsf{root}\} \\[1mm]
\delta'(\mathsf{some}_q,\sigma) &=& \left \{
    \begin{array}{l@{\quad}p{10em}}
    (-1,\mathsf{some}_q) & if $\sigma\neq \mathsf{root}$ \\
    (\pairs{0},q) & \text{otherwise}
    \end{array}
\right . \\[4mm]
\delta'(\mathsf{all}_q,\sigma) &=& \left \{
    \begin{array}{l@{\quad}p{10em}}
    (-1,\mathsf{all}_q) & \text{if } $\sigma\neq \mathsf{root}$ \\
    ([0],q) &\text{otherwise}
    \end{array}
\right .\\[4mm]
\delta'(q,\sigma) &=& \mathsf{tran}(\delta(q,\sigma)) \text{ for all }
  q\in Q \text{ and } \sigma\in\Sigma
\end{array}
$$
Here, $\mathsf{tran}(\beta)$ replaces all atoms
$(\pairs{\mathsf{root}},q)$ in $\beta$ with
$(\varepsilon,\mathsf{some}_q)$, and all atoms $([\mathsf{root}],q)$
in $\beta$ with $(\varepsilon,\mathsf{all}_q)$. The acceptance
condition $\FF'$ is identical to $\FF=\{\FF_1,\dots,\FF_k\}$, except
that all $\FF_i$ are extended with $q_r$ and $\FF_k$ is extended
with $q_0$ and all states $\mathsf{some}_q$ and $\mathsf{all}_q$.
It is not hard to see that $A'$ accepts $\pairs{T,V}$ iff $A$ accepts
 the forest encoded by $\pairs{T,V}$.
\endproof
In Section~\ref{sect:emptiness}, we shall prove the following result.
\begin{thm}\label{emptiness}
The emptiness problem for a 2GAPT $A = \genA$ with
$n$ states and index $k$ can be solved
in time
$(b+2)^{\mathcal{O}(n^3 \cdot k^2 \cdot \log k \cdot \log b^2)}$.
\end{thm}
By Lemma~\ref{tree3}, we obtain the following corollary.
\begin{cor}
\label{cor:feaempty}
  The emptiness problem for a FEA $A = \genA$ with %
  $n$ states and index $k$ can be solved
  in time
$(b+2)^{\mathcal{O}(n^3 \cdot k^2 \cdot \log k \cdot \log b^2)}$.
\end{cor}

\section{{\sc ExpTime} upper bounds for enriched $\mu$-calculi}
\label{sec:upper-bounds}

We use Theorem~\ref{emptiness} and Corollary~\ref{cor:feaempty} to
establish {\sc ExpTime} upper bounds for satisfiability in the full
graded $\mu$-calculus and the hybrid graded $\mu$-calculus.

\subsection{Full graded $\mu$-calculus}

We give a polynomial translation of formulas $\varphi$ of the full graded
$\mu$-calculus into a 2GAPT $A_\varphi$ that accepts the tree models of
$\varphi$. We can thus decide satisfiability of $\varphi$ by checking
non-emptiness of $\mathcal{L}(A_\varphi)$. There is a minor technical
difficulty to be overcome: we use Kripke structures with labeled edges, while
the trees accepted by 2GAPTs do not. This problem can be dealt with by moving
the label from each edge to the target node of the edge. For this purpose, we
introduce a new propositional symbol $p_\alpha$ for each program $\alpha$. For
a formula $\varphi$, let $\Gamma(\varphi)$ denote the set of all atomic
propositions and all propositions $p_\alpha$ such that $\aa$ is an (atomic or
inverse) program occurring in $\varphi$. The \emph{encoding} of a tree
structure $K=\langle W,R,L \rangle$ is the $2^{\Gamma(\varphi)}$-labeled tree
$\langle W,L^* \rangle$ such that
$$
L^*(w)= \{ p \in \mathsf{Prop} %
\mid w \in L(p) \} \cup \{ p_\alpha \mid \exists (v,w) \in R(\alpha)
\text{ with $w$ $\aa$-successor of $v$ in $W$}\}.
$$

For a sentence $\varphi$, we use $|\varphi|$ to denote the \emph{length} of
$\varphi$ with numbers inside graded modalities coded in binary. Formally,
$|\varphi|$ is defined by induction on the structure of $\varphi$ in a standard
way, where in particular $|\pairs{n,\aa}\psi|= |[n,\aa]\psi| = \lceil \log \ n
\rceil + 1 + |\psi|$. We say that a formula $\varphi$ \emph{counts} up to $b$
if the maximal integer in \emph{atleast} and \emph{allbut} formulas used in
$\varphi$ is $b-1$.
\begin{thm}\label{full graded to automata}
  Given a sentence $\varphi$ of the full graded $\mu$-calculus that
  counts up to $b$, we can
  construct a 2GAPT $A_\varphi$ such that $A_\varphi$
  \begin{enumerate}[\em(1)]

  \item accepts exactly the encodings of tree models of
    $\varphi$, 

  \item has $\mathcal{O}(|\varphi|)$ states, index
    $\mathcal{O}(|\varphi|)$, and counting bound $b$.

  \end{enumerate}
  The construction can be done in time $\mathcal{O}(|\varphi|)$.
\end{thm}

\proof
The automaton $A_\varphi$ verifies that $\varphi$ holds at the root of
the encoded tree.  To define the set of states, we use the
Fischer-Ladner closure $\mathsf{cl}(\varphi)$ of $\varphi$. It is
defined analogously to the Fischer-Ladner closure $\mathsf{cl}(\cdot)$
for the hybrid graded $\mu$-calculus, as given in
Section~\ref{sect:tmpfmp}.
We define $A_\varphi$ as $\pairs{2^{\Gamma(\varphi)}, b,
  \mathsf{cl}(\varphi),\delta,\varphi,\FF}$, where
the transition function $\delta$ is defined by setting, for all $\sigma \in
2^{\Gamma(\varphi)}$,
$$
\begin{array}{rcl}
  \delta(p,\sigma)&=& (p \in \sigma) \\
  \delta(\neg p,\sigma)&=& (p \not \in \sigma)\\
  \delta(\psi_1 \wedge \psi_2,\sigma)&=& (\varepsilon,\psi_1) \wedge (\varepsilon,\psi_2)\\
 \delta(\psi_1 \vee \psi_2,\sigma)&=& (\varepsilon,\psi_1) \vee (\varepsilon,\psi_2)\\
\delta(\lambda y. \psi(y),\sigma)&=& (\varepsilon,\psi(\lambda y.
\psi(y))) \\
\delta(\pairs{n,a}\psi,\sigma)&=&((-1,\psi) \wedge
(\varepsilon, p_{a^-}) \wedge (\pairs{n-1},\psi \wedge p_a)) \vee
(\pairs{n},\psi \wedge p_a)\\
\delta(\pairs{n,a^-}\psi,\sigma)&=&((-1,\psi) \wedge
(\varepsilon, p_a) \wedge (\pairs{n-1},\psi \wedge p_{a^-})) \vee
(\pairs{n},\psi \wedge p_{a^-})\\
\delta([n,a]\psi,\sigma)&=&((-1, \psi) \wedge
(\varepsilon, p_{a^-}) \wedge ([n],\psi \wedge p_a)) \vee
([n-1],\psi \wedge p_a)\\
\delta([n,a^-]\psi,\sigma)&=&((-1, \psi) \wedge
(\varepsilon, p_a) \wedge ([n],\psi \wedge p_{a^-})) \vee
([n-1],\psi \wedge p_{a^-})
\end{array}
$$
In case $n=0$, the conjuncts (resp.\ disjuncts) involving ``$n-1$''
are simply dropped in the last two lines.

The acceptance condition of $A_\varphi$ is defined in the standard way as
follows (see e.g.\ \cite{KVW00}). For a fixpoint formula $\psi \in
\mathsf{cl}(\varphi)$, the alternation level of $\psi$ is the number of
alternating fixpoint formulas one has to ``wrap $\psi$ with'' to reach a
sub-\emph{sentence} of $\varphi$.  Formally, let $\psi= \lambda y.  \psi'(y)$.
The \emph{alternation level} of $\psi$ in $\varphi$, denoted
$\mathsf{al}_\varphi(\psi)$ is defined as follows (\cite{BC96}): if $\psi$ is a
sentence, then $\mathsf{al}_\varphi(\psi)=1$.  Otherwise, let $\xi= \lambda' z.
\psi''(z)$ be the innermost $\mu$ or $\nu$ subformula of $\varphi$ that has
$\psi$ as a strict subformula.  Then, if $z$ is free in $\psi$ and $\lambda'
\neq \lambda$, we have $\mathsf{al}_\varphi(\psi) =
\mathsf{al}_\varphi(\xi)+1$; otherwise, $\mathsf{al}_\varphi(\psi) =
\mathsf{al}_\varphi(\xi)$.

Let $d$ be the maximum alternation level of (fixpoint) subformulas of
$\varphi$. Denote by $G_i$ the set of all $\nu$-formulas in
$\mathsf{cl}(\varphi)$ of alternation level $i$ and by $B_i$ the set
of all $\mu$-formulas in $\mathsf{cl}(\varphi)$ of alternation level
less than or equal to $i$. Now, define $\FF:=
\{\FF_0,\FF_1,\dots,\FF_{2d},Q\}$ with $\FF_0=\emptyset$ and for every
$1\leq i \leq d$, $\FF_{2i-1}= \FF_{2i-2} \cup B_i$ and $\FF_{2i}=
\FF_{2i-1} \cup G_i$.
Let $\pi$ be a path. By definition of $\FF$, the minimal $i$ with
$\mathsf{Inf}(\pi) \cap F_i \neq \emptyset$ determines the alternation
level and type $\lambda$ of the outermost fixpoint formula $\lambda y
. \psi(y)$ that was visited infinitely often on $\pi$.
The acceptance condition makes sure that this formula is a
$\nu$-formula. In other words, every $\mu$-formula that is visited
infinitely often on $\pi$ has a super-formula that (i)~is a
$\nu$-formula and (ii)~is also visited infinitely
often.%
\endproof
Let $\varphi$ be a sentence of the full graded $\mu$-calculus with
$\ell$ at-least subformulas.  By Theorems~\ref{tree1},
\ref{emptiness}, and \ref{full graded to automata}, the satisfiability
of $\varphi$ can be checked in time bounded by $2^{p(|\varphi|)}$
where $p(|\varphi|)$ is a polynomial (note that, in
Theorem~\ref{emptiness}, $n$, $k$, $\log \ell$, and $\log b$ are all
in $\mathcal{O}(|\varphi|)$).  This yields the desired {\sc ExpTime}
upper bound.  The lower bound is due to the fact that the
$\mu$-calculus is \exp-hard \cite{FL79}.
\begin{thm}
  The satisfiability problem of the full graded $\mu$-calculus is
  \exp-complete even if the numbers in the graded modalities are coded
  in binary.
\end{thm}

\subsection{Hybrid graded $\mu$-calculus}\label{subsec:Hybrid graded calculus}

We reduce satisfiability in the hybrid graded $\mu$-calculus to the
emptiness problem of FEAs. Compared to the reduction presented in the
previous section, two additional difficulties have to be addressed.

First, FEAs accept forests while the hybrid $\mu$-calculus has only a
\emph{quasi}-forest model property. This problem can be solved by
introducing in node labels new propositional symbols $\uparrow^a_o$
which do not occur in the input formula and represent an edge labeled
with the atomic program $a$ from the current node to the (unique) root
node labeled by nominal $o$. Let $\Theta(\varphi)$ denote the set of
all atomic propositions and nominals occurring in $\varphi$ and all
propositions $p_a$ and $\uparrow^a_o$ such that the atomic program $a$
and the nominal $o$ occur in $\varphi$.  Analogously to encodings of
trees in the previous section, the \emph{encoding} of a directed
quasi-forest structure $K=\langle W,R,L \rangle$ is the
$2^{\Theta(\varphi)}$-labeled forest $\langle W,L^* \rangle$ such that
$$
\begin{array}{rcl}
L^*(w) &=& \{ p \in \mathsf{Prop} \cup \mathsf{Nom} \mid w \in L(p) \}
\; \cup \\[1mm]
&& \{ p_a \mid \exists (v,w) \in R(a) \text{ with $w$
  successor of $v$ in $W$}\} \; \cup \\[1mm]
&&\{ \uparrow^a_o \mid \exists (w,v) \in R(a) \text{ with } L(o)=\{v\} \}.
\end{array}
$$

Second, we have to take care of the interaction between graded
modalities and the implicit edges encoded via propositions
$\uparrow^a_o$. To this end, we fix some information about the
structures accepted by FEAs already before constructing the FEA,
namely (i)~the formulas from the Fischer-Ladner closure that are
satisfied by each nominal and (ii)~the nominals that are interpreted
as the same state.
This information is provided by a so-called guess.  To introduce
guesses formally, we need to extend the Fischer-Ladner closure
$\mathsf{cl}(\varphi)$ for a formula $\varphi$ of the hybrid graded
$\mu$-calculus as follows:
$\mathsf{cl}(\varphi)$ has to satisfy the closure conditions given for the
hybrid graded $\mu$-calculus in Section~\ref{sect:tmpfmp} and, additionally,
the following:
\begin{enumerate}[$\bullet$]

\item if $\psi \in \mathsf{cl}(\varphi)$, then $\neg \psi \in
  \mathsf{cl}(\varphi)$, where $\neg \psi$ denotes the formula
  obtained from $\psi$ by dualizing all operators 
  and
  replacing every literal (i.e., atomic proposition, nominal, or
  negation thereof) with its negation.

\end{enumerate}
Let $\varphi$ be a formula with nominals $O=\{o_1,\dots,o_k\}$.  A
\emph{guess} for $\varphi$ is a pair $(t,{\sim})$ where $t$ assigns a
subset $t(o) \subseteq \mathsf{cl}(\varphi)$ to each $o \in O$ and $\sim$
is an equivalence relation on $O$ such that the following conditions
are satisfied, for all $o,o' \in O$:
\begin{enumerate}[(i)]

\item $\psi \in t(o)$ or $\neg \psi \in t(o)$ for all
  formulas $\psi \in \mathsf{cl}(\varphi)$;

\item $o \in t(o)$;

\item $o \sim o'$ implies $t(o)=t(o')$;

\item $o \not\sim o'$ implies $\neg o \in t(o')$.

\end{enumerate}
The intuition of a guess is best understood by considering the
following notion of compatibility.  A directed quasi-forest structure
$K=(W,R,L)$ is \emph{compatible} with a guess $G=(t,{\sim})$ if the following
conditions are satisfied, for all $o,o' \in O$:
\begin{enumerate}[$\bullet$]

\item $L(o)=\{ w \}$ implies that $\{ \psi \in \mathsf{cl}(\varphi) \mid
  K,w \models \psi \} = t(o)$;

\item $L(o)=L(o')$ iff $o \sim o'$.

\end{enumerate}
We construct a separate FEA $A_{\varphi,G}$ for each guess $G$ for
$\varphi$ such that $\varphi$ is satisfiable iff
$\mathcal{L}(A_{\varphi,G})$ is non-empty for some guess $G$. Since
the number of guesses is exponential in the length of $\varphi$,
we get an {\sc ExpTime} decision procedure by
constructing all of the FEAs and checking whether at least one of
them accepts a non-empty language.
\begin{thm}\label{hybrid graded to automata}
  Given a sentence $\varphi$ of the hybrid graded $\mu$-calculus that
  counts up to $b$
  and a guess $G$ %
  for $\varphi$, we can construct a FEA $A_{\varphi, G}$ such that
  \begin{enumerate}[\em(1)]

  \item if $\pairs{F,V}$ is the encoding of a directed quasi-forest
    model of $\varphi$ compatible with $G$, then $\pairs{F,V} \in
    \mathcal{L}(A_{\varphi, G})$,

  \item if $\mathcal{L}(A_{\varphi, G}) \neq \emptyset$, then there is an
    encoding $\pairs{F,V}$ of a directed quasi-forest model of
    $\varphi$ compatible with $G$ such that $\pairs{F,V} \in
    \mathcal{L}(A_{\varphi, G})$, and

  \item $A_{\varphi, G}$ has $\mathcal{O}(|\varphi|^2)$ states, index
    $\mathcal{O}(|\varphi|)$, and counting bound $b$.

  \end{enumerate}
  The construction can be done in time $\mathcal{O}(|\varphi|^2)$.
\end{thm}
\proof Let $\varphi$ be a formula of the hybrid graded $\mu$-calculus
and $G=(t,{\sim})$ a guess for $\varphi$. Assume that the nominals
occurring in $\varphi$ are $O=\{o_1,\dots,o_k\}$. For each formula
$\psi \in \mathsf{cl}(\varphi)$, atomic program $a$, and $\sigma \in
2^{\Theta(\varphi)}$, let
\begin{enumerate}[$\bullet$]

\item $\mathsf{nom}^a_\psi(\sigma) = \{ o \mid \psi \in t(o) \wedge
  {\uparrow}^a_o \in \sigma \}$;

\item $|\mathsf{nom}^a_\psi(\sigma)|^\sim$ denote the number of equivalence
  classes $C$ of $\sim$ such that some member of $C$ is contained in
  $\mathsf{nom}^a_\psi(\sigma)$.

\end{enumerate}
The automaton $A_{\varphi,G}$ verifies compatibility with $G$, and
ensures that $\varphi$ holds in some root. As its set of states, we
use
$$
Q = \mathsf{cl}(\varphi) \cup \{ q_0 \} \cup \{ \neg o_i \lor \psi,\ \mid 1
\leq i \leq k \wedge \psi \in \mathsf{cl}(\varphi) \} %
\cup \{ \mathsf{ini}_{i} \mid 1 \leq i \leq k \}.
$$
Set
$ A_{\varphi,G}= \pairs{2^{\Theta(\varphi)}, b, Q,\delta,q_0,\FF},
$
where the transition function $\delta$ and the acceptance condition
$\mathcal{F}$ is defined in the following.
For all $\sigma \in 2^{\Theta(\varphi)}$, define:
$$
\begin{array}{r@{\;}c@{\;}l}
\delta(q_0,\sigma) &=& (\langle \mathsf{root} \rangle,\varphi) \wedge
  \displaystyle\bigwedge_{1 \leq i\leq k} (\langle \mathsf{root} \rangle, o_i)
  \wedge \displaystyle\bigwedge_{1 \leq i\leq k} ([\mathsf{root}], \mathsf{ini}_i)\\
\delta(\mathsf{ini}_i,\sigma) &=& (\varepsilon, \neg o_i)  \lor
  \displaystyle\bigwedge_{\gamma \in t(o_i)} (\varepsilon,\gamma ) \\
\delta(\neg p,\sigma) &=& (p \not \in \sigma)\\
\delta(\psi_1 \wedge \psi_2,\sigma) &=& (\varepsilon, \psi_1) \wedge (\varepsilon,\psi_2)\\
\delta(\psi_1 \vee \psi_2,\sigma)&=& (\varepsilon, \psi_1) \vee (\varepsilon,\psi_2)\\
\delta(\lambda y. \psi(y),\sigma) &=& (\varepsilon,\psi(\lambda y.
\psi(y)))\\
\delta([n,a]\psi,\sigma) &=& \mathsf{false} \text{ if } |\mathsf{nom}^a_{\neg
    \psi}(\sigma)|^\sim > n\\
\delta([n,a]\psi,\sigma)&=& ([n-|\mathsf{nom}^a_{\neg
    \psi}(\sigma)|^\sim],\psi\wedge p_a) \wedge \displaystyle\bigwedge_{o \in
    \mathsf{nom}^a_\psi(\sigma)} ([\mathsf{root}], \neg o \lor \psi)  \text{ if } |\mathsf{nom}^a_{\neg
    \psi}(\sigma)|^\sim \leq n \\
\delta(\pairs{n,a}\psi,\sigma)&=&(\pairs{n-|\mathsf{nom}^a_\psi(\sigma)|^\sim},\psi
  \wedge p_a) \wedge \displaystyle\bigwedge_{o \in \mathsf{nom}^a_\psi(\sigma)}
  ([\mathsf{root}], \neg o \lor \psi)
\end{array}
$$
In the last line, the first conjunct is omitted if
$|\mathsf{nom}^a_\psi(\sigma)|^\sim > n$. The first two transition
rules check that each nominal occurs in at least one root and that the
encoded quasi-forest structure is compatible with the guess $G$.
Consider the last three rules, which are concerned with graded
modalities and reflect the existence of implicit back-edges to
nominals. The first of these rules checks for allbut formulas that are
violated purely by back-edges. The other two rules consist of two
conjuncts, each.  In the first conjunct, we subtract the number of
nominals to which there is an implicit $a$-edge and that violate the
formula $\psi$ in question. This is necessary because the $\langle
\cdot \rangle$ and $[\cdot]$ transitions of the automaton do not take
into account implicit edges.  In the second conjunct, we send a copy
of the automaton to each nominal to which there is an $a$-edge and
that satisfies $\psi$.  Observe that satisfaction of $\psi$ at this
nominal is already guaranteed by the second rule that checks
compatibility with $G$. We nevertheless need the second conjunct in
the last two rules because, without the jump to the nominal, we will
be missing paths in runs of $A_{\varphi,G}$ (those that involve an
implicit back-edge). Thus, it would not be guaranteed that these paths
satisfy the acceptance condition, which is defined below. This, in
turn, means that the evaluation of least fixpoint formulas is not
guaranteed to be well-founded. This point was missed in \cite{SV01},
and the same strategy used here can be employed to fix the
construction in that paper.

The acceptance condition of $A_{\varphi,G}$ is defined as in the case
of the full graded $\mu$-calculus: let $d$ be the maximal alternation
level of subformulas of $\varphi$, which is defined as in the case of
the full graded $\mu$-calculus. Denote by $G_i$ the set of all the
$\nu$-formulas in $\mathsf{cl}(\varphi)$ of alternation level $i$ and
by $B_i$ the set of all $\mu$-formulas in $\mathsf{cl}(\varphi)$ of
alternation depth less than or equal to $i$. Now,
$\FF=\{\FF_0,\FF_1,\dots,\FF_{2d},Q\}$, where $\FF_0=\emptyset$ and for
every $1\leq i \leq d$ we have $\FF_{2i-1}= \FF_{2i-2} \cup B_i$, and
$\FF_{2i}= \FF_{2i-1} \cup G_i$.

It is standard to show that if $\pairs{F,V}$ is the encoding of a directed
quasi-forest model $K$ of $\varphi$ compatible with $G$, then $\pairs{F,V} \in
\mathcal{L}(A_{\varphi,G})$. Conversely, let $\pairs{F,V}\in
\mathcal{L}(A_{\varphi,G})$. If $\pairs{F,V}$ is \emph{nominal unique}, i.e.,
if every nominal occurs only in the label of a single root, it is not hard to
show that $\pairs{F,V}$ is the encoding of a directed quasi-forest model $K$ of
$\varphi$ compatible with $G$. However, the automaton $A_{\varphi,G}$ does not
(and cannot) guarantee nominal uniqueness.  To establish Point~(2) of the
theorem, we thus have to show that whenever $\mathcal{L}(A_{\varphi,G}) \neq
\emptyset$, then there is an element of $\mathcal{L}(A_{\varphi,G})$ that is
nominal unique.

Let $\pairs{F,V}\in \mathcal{L}(A_{\varphi,G})$. From $\pairs{F,V}$, we extract
a new forest $\pairs{F',V'}$ as follows: Let $r$ be a run of $A_{\varphi,G}$ on
$\pairs{F,V}$. Remove all trees from $F$ except those that occur in $r$ as
witnesses for the existential root transitions in the first transition rule.
Call the modified forest $F'$. Now modify $r$ into a run $r'$ on $F'$: simply
drop all subtrees rooted at nodes whose label refers to one of the trees that
are present in $F$ but not in $F'$.  Now, $r'$ is a run on $F'$ because (i)~the
only existential root transitions are in the first rule, and these are
preserved by construction of $F'$ and $r'$; and (ii)~all universal root
transitions are clearly preserved as well. Also, $r'$ is accepting because
every path in $r'$ is a path in~$r$.  Thus, $\pairs{F',V'}\in
\mathcal{L}(A_{\varphi,G})$ and it is easy to see that $\pairs{F',V'}$ is
nominal unique.
\endproof
Combining Theorems~\ref{tree2}, Corollary~\ref{cor:feaempty}, and
Theorem~\ref{hybrid graded to automata}, we obtain an {\sc
  ExpTime}-upper bound for the hybrid graded $\mu$-calculus. Again,
the lower bound is from \cite{FL79}.
\begin{thm}
  The satisfiability problems of the full graded $\mu$-calculus and
  the hybrid graded $\mu$-calculus are \exp-complete even if the
  numbers in the graded modalities are coded in binary.
\end{thm}

\section{The Emptiness Problem for 2GAPTs}
\label{sect:emptiness}

We prove Theorem~\ref{emptiness} and thus show that the emptiness
problem of 2GAPTs can be solved in {\sc ExpTime}. The proof is by a
reduction to the emptiness problem of graded nondeterministic parity
tree automata (GNPTs) as introduced in \cite{KSV02}.

\subsection{Graded nondeterministic parity tree automata}

We introduce the graded nondeterministic parity tree automata (GNPTs)
of \cite{KSV02}.
For $b > 0$, a \emph{$b$-bound} is a pair in
$$
  B_b = \{ (>,
0), (\leq, 0), (>, 1), (\leq, 1), \ldots, (>, b), (\leq,
b)\}.
$$
For a set $X$, a subset $P$ of $X$, and a (finite or
infinite) word $t = x_1x_2\cdots \in X^* \cup X^\omega$, the
\emph{weight} of $P$ in $t$, denoted $\mathsf{weight}(P, t)$, is the
number of occurrences of symbols in $t$ that are members of $P$. That
is, $\mathsf{weight}(P,t) = |\{i : x_i \in P\}|$.  For example,
$\mathsf{weight}(\{1,2\},1241) = 3$. We say that $t$ satisfies a
$b$-bound $(>, n)$ with respect to $P$ if $\mathsf{weight}(P, t) > n$,
and $t$ satisfies a $b$-bound $(\leq, n)$ with respect to $P$ if
$\mathsf{weight}(P,t) \leq n$.

For a set $Y$, we use $B(Y)$ to denote the set of all Boolean
formulas over atoms in $Y$. Each formula $\theta \in B(Y)$ induces a
set $\mathsf{sat}(\theta)\subseteq 2^Y$ such that $x \in \mathsf{sat}(\theta)$ iff $x$
satisfies $\theta$. For an integer $b\geq 0$, a \emph{$b$-counting
constraint} for $2^Y$ is a relation $C \subseteq B(Y) \times B_b$.
For example, if $Y=\{y_1, y_2, y_3\}$, then we can have
$$C= \{\pairs{y_1 \vee \neg y_2, (\leq, 3)} , \pairs{y_3, (\leq, 2)},
\pairs{y_1 \wedge y_3, (>, 1)}\}.$$
A word $t = x_1 x_2 \cdots \in (2^Y)^* \cup (2^Y)^\omega$ satisfies
the $b$-counting constraint $C$ if for all $\pairs{\theta,\xi} \in C$,
the word $t$ satisfies $\xi$ with respect to $\mathsf{sat}(\theta)$,
that is, when $\theta$ is paired with $\xi=(>,n)$, at least $n + 1$
occurrences of symbols in $t$ should satisfy $\theta$, and when
$\theta$ is paired with $\xi=(\leq, n)$, at most $n$ occurrences
satisfy $\theta$.
For example, the word $t_1=\emptyset\{y_1\}\{y_2\}\{y_1, y_3\}$ does
not satisfy the constraint $C$ above, as the number of sets in $t_1$
that satisfies $y_1 \wedge y_3$ is one. On the other hand, the word
$t_2=\{ y_2 \}\{y_1\}\{y_1, y_2, y_3\}\{y_1, y_3\}$ satisfies $C$.
Indeed, three sets in $t_2$ satisfy $y_1 \vee \neg y_2$, two sets
satisfy $y_3$, and two sets satisfy $y_1 \wedge y_3$.

We use ${\mathcal{C}}(Y,b)$ to denote the
set of all 
$b$-counting constraints for $2^Y$.  We assume that the integers in
constraints are coded in binary.

We can now define \emph{graded nondeterministic parity tree automata}
(GNPTs, for short). A GNPT is a tuple ${A = \langle \Sigma, b, Q,
  \delta, q_0, \FF \rangle}$ where $\Sigma, \ b, \ q_0$, and $\FF$ are
as in 2GAPT, $Q \subseteq 2^Y$ is the set of states (i.e., $Q$ is
encoded by a finite set of variables), and $\delta:Q \times \Sigma
\rightarrow {\mathcal{C}}(Y,b)$ maps a state and a letter to a
$b$-counting constraint $C$ for $2^Y$ such that the cardinality of
$C$ is bounded by $\log |Q|$.
For defining runs, we introduce an
additional notion. Let $x$ be a node in a $\Sigma$-labeled tree
$\langle T, V \rangle$, and let $x \cdot i_1, x \cdot i_2 , \dots$
be the (finitely or infinitely many) successors of $x$ in $T$, where $i_j < i_{j+1}$ (the actual ordering is not important, but has to be fixed). Then we use
$\mathsf{lab}(x)$ to denote the (finite or infinite) word of labels
induced by the successors, i.e., $\mathsf{lab}(x)=V(x \cdot i_1)
V(x \cdot i_2) \cdots $.
Given a GNPT $A$, a \emph{run} of $A$ on a
$\Sigma$-labeled tree $\pairs{T, V}$ rooted in $z$ is then a $Q$-labeled
tree $\pairs{T, r}$ such that
\begin{enumerate}[$\bullet$]
\item $r(z) = q_0$ and
\item for every $x \in T$, %
$\mathsf{lab}(x)$
 satisfies $\delta(r(x), V(x))$.
\end{enumerate}
Observe that, in contrast to the case of alternating automata, the input tree
$\pairs{T,V}$ and the run $\pairs{T,r}$ share the component $T$. The run
$\pairs{T, r}$ is \emph{accepting} if all its infinite paths satisfy the parity
acceptance condition.
A GNPT \emph{accepts} a tree iff there exists an accepting run of the
automaton on the tree. We denote by $\mathcal{L}(A)$ the set of all
$\Sigma$-labeled trees that $A$ accepts.

We need two special cases of GNPT: \forrall automata and \safety automata. In
\forrall automata, for each $q \in Q$ and $\sigma \in \Sigma$ there is a $q'
\in Q$ such that $\delta(q , \sigma) = \{\pairs{(\neg
\theta_{q'}),(\leq,0)}\}$, where $\theta_{q'} \in B(Y)$ is such that
$\mathsf{sat}(\theta_{q'}) = \{\{q'\}\}$. Thus, a \forrall automaton is very
similar to a (non-graded) deterministic parity tree automaton, where the
transition function maps $q$ and $\sigma$ to $\pairs{q',\dots, q'}$ (and the
out-degree of trees is not fixed). In \safety automata, there is no acceptance
condition, and all runs are accepting. Note that this does not mean that
\safety automata accept all trees, as it may be that on some trees the
automaton does not have a run at all.

We need two simple results concerning GNPTs.
The following has been stated (but not proved) already
in~\cite{KSV02}.
\begin{lem}
\label{intersection}
Given a \forrall GNPT $A_1$ with $n_1$ states and index $k_1$, and a
\safety GNPT $A_2$ with $n_2$ states and counting bound $b_2$, we can
define a GNPT $A$ with $n_1 n_2$ states, index $k_1$, and counting
bound $b_2$, such that ${\mathcal{L}}(A) = {\mathcal{L}}(A_1) \cap
{\mathcal{L}}(A_2)$.
\end{lem}
\begin{proof}
  We can use a simple product construction. Let
  $A_i=(\Sigma,b_i,Q_i,\delta_i,q_{0,i}, \FF^{(i)})$ with $Q_i
  \subseteq 2^{Y_i}$ for $i \in \{1,2\}$. Assume w.l.o.g.\ that $Y_1
  \cap Y_2 = \emptyset$. We define
  $A=(\Sigma,b_2,Q,\delta,(q_{0,1} \cup q_{0,2}),\FF)$, where
  \begin{enumerate}[$\bullet$]

  \item $Q = \{ q_1 \cup q_2 \mid q_1 \in Q_1 \text{ and } q_2 \in Q_2 \} \subseteq 2^Y$, where $Y=Y_1 \uplus Y_2$;

  \item for all $\sigma \in \Sigma$ and $q = q_1 \cup q_2 \in Q$ with
    $\delta_1(q_1,\sigma)= \{\pairs{(\neg \theta_q),(\leq,0)}\}$ and
    $\delta_2(q_2,\sigma)=C$, we set $\delta(q,\sigma)= C \cup
    \{\pairs{(\neg \theta'_q),(\leq,0)}\}$, where $\theta'_{q} \in
    B(Y)$ is such that $\mathsf{sat}(\theta'_{q}) = \{q' \in Q \mid q' \cap Q_1 = q\}$;

  \item $\FF = \{ \FF_1, \dots, \FF_k \}$ with $\FF_{i} = \{ q \in Q \mid
    q \cap Q_1 \in \FF^{(1)}_{i}\}$
    if $\FF^{(1)} = \{ \FF^{(1)}_{1},\dots,\FF^{(1)}_{k} \}$.

  \end{enumerate}
  It is not hard to check that $A$ is as required.
\end{proof}
The following result can be proved by an analogous product
construction.
\begin{lem}
\label{intersectiontwo} Given \safety GNPTs $A_i$ with $n_i$ states
and counting bounds $b_i$, $i \in \{1,2\}$, we can define a \safety
GNPT $A$ with $n_1 n_2$ states and counting bound $b=\mathsf{max}\{b_1 ,
b_2\}$ such that ${\mathcal{L}}(A) = {\mathcal{L}}(A_1) \cap
{\mathcal{L}}(A_2)$.
\end{lem}

\subsection{Reduction to Emptiness of GNPTs}
\label{sec:2gaptempty}

We now show that the emptiness problem of 2GAPTs can be reduced to the
emptiness problem of GNPTs that are only exponentially larger.
Let $A=\genA$ be a 2GAPT. We recall that $\delta$ is a function from $Q \times
\Sigma$ to $B^{+}(D_b^- \times Q)$, with $D_b^- := \ppairs{b} \cup [[b]] \cup
\{-1, \varepsilon\}$. A \emph{strategy tree} for $A$ is a $2^{Q \times D_b^-
\times Q}$-labeled tree $\pairs{T,\strat}$.  Intuitively, the purpose of a
strategy tree is to guide the automaton $A$ by pre-choosing transitions that
satisfy the transition relation. For each label $w=\strat(x)$, we use $\head(w)
= \{q \mid (q, c, q') \in w \}$ to denote the set of states for which $\strat$
chooses transitions at $x$. Intuitively, if $A$ is in state $q \in \head(w)$,
\strat tells it to execute the transitions $\{(c,q') \mid (q,c,q') \in w \}$.
In the following, we usually consider only the \strat part of a strategy tree.
Let $\pairs{T,V}$ be a $\Sigma$-labeled tree and $\pairs{T,\strat}$ a strategy
tree for $A$, based on the same $T$.  Then \strat is a \emph{strategy for $A$
on
  $V$} if for all nodes $x \in T$ and all states $q \in Q$, we have:
\begin{enumerate}[(1)]

\item $\delta(q_0,V(\mathsf{root}(T)))=\mathsf{true}$ or $q_0 \in
  \head(\strat(\mathsf{root}(T)))$;

\item if $q \in \head(\strat(x))$, then the set $\{(c,q') : (q,c,q')
  \in \strat(x)\}$ satisfies $\delta(q,V(x))$,

\item if $(q,c,q') \in \strat(x)$ with $c \in \{-1,\varepsilon\}$,
  then (i)~$x \cdot c$ is defined and (ii)~$\delta(q',V(x \cdot
  c))=\mathsf{true}$ or $q' \in \head(\strat(x \cdot c))$.

\end{enumerate}
If $A$ is understood, we simply speak of a strategy on $V$.

\begin{exa}
\label{exa1} Let $A=\genA$ be a 2GAPT such that $\Sigma=\{a,b,c\}$, $Q=\{q_0,
q_1, \allowbreak q_2, q_3\}$, and $\delta$ is such that
$\delta(q,a)=(\pairs{0},q_1) \vee (\pairs{0},q_3)$ for $q \in \{q_0, q_2\}$,
and $\delta(q_1,b)=((-1,q_2) \wedge (\pairs{1},q_3)) \vee ([1],q_1)$.  Consider
the trees depicted in Figure~\ref{figure:strategy}. From left to right, the
first tree $\pairs{T,V}$ is a fragment of the input tree, the second tree is a
fragment of a run $\pairs{T_r,r}$ of $A$ on $\pairs{T,V}$, and the
third tree is a fragment of a strategy tree suggesting this run.  In a
label $\langle w,a \rangle$ of the input tree, $w$ is the node name
and $a \in \Sigma$ the label in the tree.  In the run and strategy
tree, only the labels are given, but not the node names.
\end{exa}

\begin{figure}[t]
\footnotesize

    \begin{picture}(150,150)

     \thicklines
     \put(-30,115){\makebox(0,0)[t] {$\pairs{1,a}$}}
     \put(10,114){\makebox(0,0)[t] 
     }
     \put(-32,102){\vector(-1,-1){28}}
     \put(-68,70){\makebox(0,0)[t]{$\pairs{11,b}$}}
     \put(-34,70){\makebox(0,0)[t] 
     }
     \put(-28,102){\vector(1,-1){26}}
     \put(5,70){\makebox(0,0)[t]{$\pairs{12,a}$}}
     \put(40,70){\makebox(0,0)[t] 
     }
     \put(-65,60){\vector(-1,-1){25}}
     \put(-65,60){\vector(1,-1){25}}
     \put(-90,30){\makebox(0,0)[t]{$\pairs{111,b}$}}
     \put(-40,30){\makebox(0,0)[t]{$\pairs{112,a}$}}
     \put(5,30){\makebox(0,0)[t] 
     }

     \put(120,120){\makebox(0,0)[t] {$(1,q_0$)}}
     \put(120,110){\vector(-1,-2){10}}
     \put(105,85){\makebox(0,0)[t]{$(11, q_1)$}}
     \put(105,75){\vector(-1,-1){25}}
     \put(105,75){\vector(0,-1){25}}
     \put(105,75){\vector(1,-1){25}}
     \put(75,45){\makebox(0,0)[t]{$(1, q_2)$}}
     \put(105,45){\makebox(0,0)[t]{$(111, q_3)$}}
     \put(140,45){\makebox(0,0)[t]{$(112, q_3)$}}
     \put(80,35){\vector(0,-1){20}}
     \put(80,15){\makebox(0,0)[t] {$(12,q_3$)}}

     \put(250,115){\makebox(0,0)[t] {$(q_0,\pairs{0},q_1)$,$(q_2,\pairs{0},q_3)$}}
     \put(250,100){\vector(-1,-1){25}}
     \put(255,100){\vector(1,-1){25}}
     \put(220,70){\makebox(0,0)[t] {$(q_1,-1,q_2)$, $(q_1,\pairs{1},q_3)$}}
     \put(220,60){\vector(-1,-1){25}}
     \put(220,60){\vector(1,-1){25}}

    \end{picture}

\caption{A fragment of an input tree, a corresponding run, and its
strategy tree.} \label{figure:strategy}
\end{figure}

Strategy trees do not give full information on how to handle
transitions $(\langle n \rangle,q)$ and $([n],q)$ as they do not say
which successors should be used when executing them.  This is
compensated by \emph{promise trees}. A promise tree for $A$ is a $2^{Q
  \times Q}$-labeled tree $\pairs{T,\prom}$.  Intuitively, if a run
that proceeds according to $\prom$ visits a node $x$ in state $q$ and
chooses a move $(\pairs{n},q')$ or $([n],q')$, then the successors $x
\cdot i$ of $x$ that inherit $q'$ are those with $(q,q') \in \prom(x
\cdot i)$.
Let $\pairs{T,V}$ be a $\Sigma$-labeled tree, $\strat$ a strategy on
$V$, and $\pairs{T,\prom}$ a promise tree.  We call \prom a
\emph{promise for $A$ on} \strat if the states promised to be visited
by $\prom$ satisfy the transitions chosen by $\strat$, i.e., for every
node $x \in T$,
the following hold:
\begin{enumerate}[(1)]

\item for every $(q, \pairs{n}, q') \in \strat(x)$, there is a subset
  $M \subseteq \mathsf{succ}(x)$ of cardinality $n+1$ such
  that each $y \in M$ satisfies $(q,q') \in \prom(y)$;

\item for every $(q, [n], q') \in \strat(x)$, there is a subset $M
  \subseteq \mathsf{succ}(x)$ of cardinality $n$ such that
  each $y \in \mathsf{succ}(x) \setminus M$ satisfies $(q,q') \in
  \prom(y)$;

\item if $(q,q') \in \prom(x)$, then $\delta(q',V(x))=\mathsf{true}$
  or $q' \in \head(\strat(x))$.

\end{enumerate}

Consider a $\Sigma$-labeled tree $\pairs{T, V}$, a strategy $\strat$
on $V$, and a promise $\prom$ on $\strat$. An infinite
sequence of pairs $(x_0,q_0), (x_1,q_1) \ldots$ is a \emph{trace}
induced by $\strat$ and $\prom$ if $x_0$ is the root of $T$, $q_0$ is
the initial state of $A$ and, for each $i\geq 0$, one of the following
holds:
\begin{enumerate}[$\bullet$]

\item there is $(q_i, c, q_{i+1}) \in \strat(x_i)$ with $c =-1$ or $c =\varepsilon$,
$x_i \cdot c$ defined, and $x_{i+1}=x_i \cdot c$;

\item $\strat(x_i)$ contains $(q_i,\pairs{n},q_{i+1})$ or
$(q_i,[n],q_{i+1})$, there exists $j \in \Nat$ with $x_{i+1}= x_i
\cdot j \in T$, and $(q_i,q_{i+1}) \in \prom(x_{i+1})$.

\end{enumerate}
Let $\FF=\{ \FF_1,\dots,\FF_k \}$.  For each state $q \in Q$, let
$\mathsf{index}(q)$ be the minimal $i$ such that $q \in \FF_i$. For a trace
$\pi$, let $\mathsf{index}(\pi)$ be the minimal index of states that occur
infinitely often in $\pi$.  Then, $\pi$ \emph{satisfies} $\FF$ if it has even
index. The strategy $\strat$ and promise $\prom$ are \emph{accepting} if all
the traces induced by $\strat$ and $\prom$ satisfy $\FF$.

In \cite{KSV02}, it was shown that a necessary and sufficient
condition for a tree $\pairs{T,V}$ to be accepted by a one-way GAPT is
the existence of a strategy $\strat$ on $V$ and a promise $\prom$ on
$\strat$ that are accepting. We establish the same result for the case
of 2GAPTs.
\begin{lem}\label{good for accepting}
  A 2GAPT $A$ accepts $\pairs{T, V}$ iff there exist a strategy
  $\strat$ for $A$ on $V$ and a promise $\prom$ for $A$ on $\strat$
  that are accepting.
\end{lem}
\proof Let $A=\genA$ be a 2GAPT with ${\FF = \{\FF_1, \ldots ,
\FF_k \}}$, and let $\pairs{T,V}$ be the input tree. Suppose first
that $A$ accepts $\pairs{T,V}$.
Consider a two-player game on $\Sigma$-labeled trees, Protagonist vs.\
Antagonist, such that Protagonist is trying to show that $A$ accepts
the tree, and Antagonist is challenging that. A configuration of the
game is a pair in $T \times Q$. The initial configuration is
$(\mathsf{root}(T), q_0)$.  Consider a configuration $(x,q)$.
Protagonist is first to move and chooses a set $P_1=\{(c_1, q_1),
\ldots, (c_m, q_m)\} \subseteq D_b^- \times Q$ that satisfies
$\delta(q, V (x))$. If $\delta(q, V(x)) = \mathsf{false}$, then
Antagonist wins immediately. If $P_1$ is empty, Protagonist wins
immediately.  Antagonist responds by choosing an element $(c_i, q_i)$
of $P_1$. If $c_i \in \{-1,\varepsilon\}$, then the new configuration
is $(x \cdot c_i, q_i)$. If $x \cdot c_i$ is undefined, then
Antagonist wins immediately. If $c_i = \pairs{n}$, Protagonist chooses
a subset $M \subseteq \mathsf{succ}(x)$ of cardinality $n+1$,
Antagonist wins immediately if there is no such subset and otherwise
responds by choosing an element $y$ of $M$. Then, the new configuration
is $(y,q_i)$. If $c_i = [n]$, Protagonist chooses a subset $M
\subseteq \mathsf{succ}(x)$ of cardinality at most $n$, Antagonist
wins immediately if there is no such subset and otherwise responds by
choosing an element $y$ of $\mathsf{succ}(x) \setminus M$. Protagonist
wins immediately if there is no such element. Otherwise, the new
configuration is $(y,q_i)$.

Consider now an infinite game $Y$, that is, an infinite sequence of
immediately successive game configurations. Let $\mathsf{Inf}(Y)$ be
the set of states in $Q$ that occur infinitely many times in $Y$.
Protagonist wins if there is an even $i>0$ for which $\mathsf{Inf}(Y)
\cap \FF_i \neq \emptyset$ and $\mathsf{Inf}(Y) \cap \FF_{j} =
\emptyset$ for all $j < i$. It is not difficult to see that a winning
strategy of Protagonist against Antagonist is essentially a
representation of a run of $A$ on $\pairs{T,V}$ and vice versa. Thus,
such a winning strategy exists iff $A$ accepts this tree. The
described game meets the conditions in \cite{Jut95}. It follows that
if Protagonist has a winning strategy, then it has a memoryless
strategy, i.e., a strategy whose moves do not depend on the history of
the game, but only on the current configuration.

Since we assume that $A$ accepts the input tree $\pairs{T,V}$,
Protagonist has a memoryless winning strategy on $\pairs{T,V}$.  This
winning strategy can be used to build a strategy $\strat$ on $V$ and a
promise $\prom$ on $\strat$ in the following way. For each $x \in T$,
$\strat(x)$ and $\prom(x)$ are the smallest sets such that, for all
configurations $(x,q)$ occurring in Protagonist's winning strategy, if
Protagonist chooses a subset $P_1=\{(c_1, q_1), \ldots, (c_m, q_m)\}$
of $D_b^- \times Q$ in the winning strategy, then we have
\begin{enumerate}[$\bullet$]

\item[(i)] $\{q\} \times P_1 \subseteq \strat(x)$ and

\item[(ii)] for each atom $(c_i, q_i)$ of $P_1$ with $c_i = \pairs{n}$
  (resp.\ $c_i = [n]$) if $M = \{y_1, \ldots, y_{n+1}\}$ (resp.\ $M =
  \{y_1, \ldots, y_n\}$) is the set of successors chosen by
  Protagonist after Antagonist has chosen $(c_i,q_i)$, then we have
  $(q,q_i) \in \prom(y)$ for each $y \in M$ \mbox{(resp.\ for each $y \in
  \mathsf{succ}(x)\setminus M$)}.

\end{enumerate}
Using the definition of games and the construction of \strat, it is
not hard to show that $\strat$ is indeed a strategy on $V$.
Similarly, it is easy to prove that $\prom$ is a promise on
$\strat$.
Finally, it follows from the definition of wins of Protagonist that $\strat$
and $\prom$ are accepting.

Assume now that there exist a strategy $\strat$ on $V$ and a promise
$\prom$ on $\strat$ that are accepting.  Using $\strat$ and
$\prom$, it is straightforward to inductively build an accepting run
$\pairs{T_r,r}$ of $A$ on $\pairs{T,V}$:
%
\begin{enumerate}[$\bullet$]

\item start with introducing the root $z$ of $T_r$, and set
  $r(z)=(\mathsf{root}(T),q_0)$;

\item if $y$ is a leaf in $T_r$ with $r(y)=(x,q)$ and $\delta(q,V(x))
  \neq \mathsf{true}$, then do the following for all $(q,c,q') \in
  \strat(x)$:
\begin{enumerate}[$-$]
\item If $c =-1$ or $c =\varepsilon$, then add a fresh successor
  $y \cdot j$ to $y$ in $T_r$ and set $r(y \cdot j)=(x \cdot c,q')$;

\item If $c = \pairs{n}$ or $c = [n]$, then for each $j \in \Nat$ with
$(q,q') \in \prom(x \cdot j)$, add a fresh successor $y \cdot j'$
to $y$ in $T_r$  and set $r(y \cdot j') = (x \cdot j, q')$.
\end{enumerate}

\end{enumerate}
By Condition~(3) of strategy trees, $y \cdot j$ is defined in the
induction step. Using the properties of strategies on $V$ and of
promises on $\strat$, it is straightforward to show that
$\pairs{T_r,r}$ is a run. It thus remains to prove that
$\pairs{T_r,r}$ is accepting.  Let $\pi$ be a path in $\pairs{T_r,r}$.
By definition of traces induced by $\strat$ and $\prom$, the labeling
of $\pi$ is a trace induced by $\strat$ and $\prom$. Since $\strat$
and $\prom$ are accepting, so is $\pi$.  \endproof
Strategy and promise trees together serve as a witness for acceptance
of an input tree by a 2GAPT that, in contrast to a run $\langle
T_r,r\rangle$, has the same tree structure as the input tree. To
translate 2GAPTs into GNPTs, we still face the problem that traces in
strategies and promises can move both up and down. To restrict
attention to unidirectional paths, we extend to our setting the notion
of annotation as defined in \cite{Var98}. Annotations allow
decomposing a trace of a strategy and a promise into a downward part
and several finite parts that are \emph{detours}, i.e., divert from
the downward trace and come back to the point of diversion.

Let $A = \genA$ be a 2GAPT.
An \emph{annotation tree} for $A$ is a $2^{Q \times \{1, \ldots, k\}
  \times Q}$-labeled tree $\pairs{T,\ann}$.  Intuitively, $(q, i,q')
\in \ann(x)$ means that from node $x$ and state $q$, $A$ can make a
detour and comes back to $x$ with state $q'$ such that $i$ is the
smallest index of all states that have been seen along the detour.
Let $\pairs{T,V}$ be a $\Sigma$-labeled tree, \strat a strategy on
$V$, $\prom$ a promise on $\strat$, and $\pairs{T,\ann}$ an annotation
tree. We call $\ann$ an annotation for $A$ on \strat and \prom if for
every node $x \in T$,
the following conditions are satisfied:
\begin{enumerate}[(1)]{\mathsurround=1 pt
\item If $(q,\varepsilon, q') \in \strat(x)$ then $(q,
  \mathsf{index}(q'),q') \in \ann(x)$;

\item if $(q, j', q') \in \ann(x)$ and $(q', j'', q'') \in \ann(x)$,
then $(q, \mathsf{min}(j',j''), q'') \in \ann(x)$;

\item if (i)~$x = y \cdot i$, (ii)~$(q, -1, q') \in \strat(x)$,
  (iii)~$(q', j, q'') \in \ann(y)$ or $q'=q''$ with
  $\mathsf{index}(q')=j$, (iv)~$(q'', \pairs{n},q''')\in\strat(y)$ or
  $(q'', [n],q''')\in\strat(y)$, and (v)~$(q'', q''') \in \prom(x)$,
  then ${(q, \mathsf{min}(\mathsf{index}(q'), j, \mathsf{index}(q''')
    ), \linebreak q''')\in \ann(x)}$;

\item if (i)~$y = x \cdot i$, (ii)~$(q, \pairs{n},q')\in\strat(x)$ or $(q,
  [n],q')\in\strat(x)$, (iii)~$(q, q') \in \prom(y)$, (iv)~$(q',
  j, q'') \in \ann(y)$ or $q'=q''$ with $\mathsf{index}(q')=j$, and
  (v)~$(q'', -1, q''') \in \strat(y)$, then $(q,
  \mathsf{min}(\mathsf{index}(q'), j, \mathsf{index}(q''')),q''') \in
  \ann(x)$.
}
\end{enumerate}

\begin{exa}
  Reconsider the 2GAPT $A=\genA$ from Example~\ref{exa1}, as well as
  the fragments of the input tree $\langle T,V \rangle$ and the
  strategy $\strat$ on $\langle T,V \rangle$ depicted in
  Figure~\ref{figure:strategy}. Assume that there is a promise \prom on
  $\strat$ with $(q_0,q_1) \in \prom(11)$ telling the automaton that
  if it executes $(\langle0\rangle, q_1)$ in state $q_0$ at node 1, it
  should send a copy in state $q_1$ to node 11. Using $\strat(1)$ and
  Condition~(4) of annotations, we can now deduce that, in any
  annotation $\ann$ on $\strat$ and $\prom$, we have $(q_0,j,q_2) \in \ann(1)$
  with $j$ the minimum of the indexes of $q_0$, $q_1$, and $q_2$.
\end{exa}

Given an annotation tree $\pairs{T,\ann}$ on $\strat$ and
$\prom$, a \emph{downward trace} $\pi$ induced by $\strat$,
$\prom$, and $\ann$ is a sequence $(x_0, q_0, t_0), (x_1, q_1, t_1),
\ldots$ of triples, where $x_0=\mathsf{root}(T)$, $q_0$ is the initial
state of $A$, and for each $i \geq
0$, %
one of the following holds:
\begin{itemize}

\item[\llap{($\dagger$)}] $t_i$ is $(q_i, c, q_{i+1}) \in \strat(x_i)$ for some $c \in
  [[b]] \cup [\pairs{b}]$, $(q_i, q_{i+1}) \in \prom(x_i \cdot d)$ for
  some $d \in \Nat$, and $x_{i+1} = x_i \cdot d$

\item[\llap{($\ddagger$)}] $t_i$ is $(q_i, d, q_{i+1}) \in \ann(x_i)$ for some $d \in \{1,
  \ldots, k\}$, and $x_{i+1} = x_i$.

\end{itemize}
In the first case, $\mathsf{index}(t_i)$ is the minimal $j$ such that
$q_{i+1} \in \FF_j$ and in the second case, $\mathsf{index}(t_i)=d$.
For a downward trace $\pi$, $\mathsf{index}(\pi)$ is the minimal
$\mathsf{index}(t_i)$ for all $t_i$ occurring infinitely often in
$\pi$. Note that a downward trace $\pi$ can loop indefinitely at a
node $x \in T$ when, from some point $i \geq 0$ on, all the $t_j$, $j
\geq i$, are elements of $\ann$ (and all the $x_j$ are $x$).  We
say that a downward trace $\pi$ \emph{satisfies} $\FF= \{
\FF_1,\dots,\FF_k \}$ if $\mathsf{index}(\pi)$ is even. Given a
strategy \strat, a promise \prom on \strat, an annotation \ann on
\strat and \prom, we say that $\ann$ is \emph{accepting} if all
downward traces induced by $\strat$, $\prom$, and $\ann$ satisfy
$\FF$.
\begin{lem}\label{annotaion must be accepting}
  A 2GAPT $A$ accepts $\pairs{T, V}$ iff there exist a strategy
  $\strat$ for $A$ on $V$, a promise $\prom$ for $A$ on
  $\strat$, and an annotation $\ann$ for $A$ on
  $\strat$ and $\prom$ such that $\ann$ is accepting.
\end{lem}
\proof Suppose first that $A$ accepts $\pairs{T, V}$. By
Lemma~\ref{good for accepting}, there is a strategy $\strat$ on $V$
and a promise $\prom$ on $\strat$ which are accepting.  By definition
of annotations on \strat and \prom, it is obvious that there exists a
unique smallest annotation $\ann$ on $\strat$ and $\prom$ in the sense
that, for each node $x$ in $T$ and each annotation $\ann'$, we have
$\ann(x) \subseteq \ann'(x)$. We show that $\ann$ is accepting. Let
$\pi = (x_0, q_0, t_0), (x_1, q_1, t_1), \ldots$ be a downward trace
induced by $\strat$, $\prom$, and $\ann$. It is not hard to construct
a trace $\pi'=(x'_0,q'_0),(x'_1,q'_1),\dots$ induced by $\strat$ and
$\prom$ that is accepting iff $\pi$ is: first expand $\pi$ by
replacing elements in $\pi$ of the form $(\ddagger)$ with the
detour asserted by $\ann$, and then project $\pi$ on the first two
components of its elements. Details are left to the reader.

Conversely, suppose that there exist a strategy $\strat$ on $V$, a promise
$\prom$ on $\strat$, and an annotation $\ann$ on $\strat$ and $\prom$ such that
$\ann$ is accepting.  By Lemma~\ref{good for
  accepting}, it suffices to show that $\strat$ and $\prom$ are
accepting.  Let $\pi=(x_0,q_0),(x_1,q_1),\dots$ be a trace induced by
$\strat$ and $\prom$. It is possible to construct a downwards trace
$\pi'$ induced by $\strat$, $\prom$, and $\ann$ that is accepting iff
$\pi$ is: whenever the step from $(x_i,q_i)$ to $(x_{i+1},q_{i+1})$ is
such that $x_{i+1}=x_i \cdot c$ for some $c \in \Nat$, the definition
of traces induced by \strat and \prom ensures that there is a $t_i =
(q_i, c, q_{i+1}) \in \strat(x_i)$ such that the conditions from
$(\dagger)$ are satisfied; otherwise, we consider the maximal
subsequence $(x_i,q_i),\dots,(x_j,q_j)$ of $\pi$ such that $x_j = x_i
\cdot c$ for some $c \in \Nat$, and replace it with
$(x_i,q_i),(x_j,q_j)$. By definition of annotations, there is
$t_i=(q_i, d, q_{i+1}) \in \ann(x_i)$ such that the conditions from
($\ddagger$) are satisfied. Again, we leave details to the reader.
\endproof

In the following, we combine the input tree, the strategy, the
promise, and the annotation into one tree $\pairs{T, (V, \strat,
  \prom, \ann)}$. The simplest approach to representing the strategy
as part of the input tree is to additionally label the nodes of the input tree
with an element of $2^{Q \times D_b^- \times Q}$. However, we can achieve
better bounds if we represent strategies more compactly. Indeed, it suffices to
store for every pair of states $q,q' \in Q$, at most four different tuples
$(q,c,q')$: two for $c \in \{ \varepsilon, -1 \}$ and two for the minimal $n$
and maximal $n'$ such that $(q,[n],q'),(q,\pairs{n'},q') \in \strat(y)$.  Call
the set of all representations of strategies $L_\strat$. We can now define the
alphabet of the combined trees. Given an alphabet $\Sigma$ for the input tree,
let $\Sigma'$ denote the extended signature for the combined trees, i.e.,
$\Sigma'=\Sigma \times L_\strat \times 2^{Q \times Q} \times 2^{Q \times
\{1,\ldots, k\}\times Q}$.
\begin{thm}\label{translation}
  Let $A$ be a 2GAPT running on $\Sigma$-labeled trees with $n$
  states, index $k$ and counting bound $b$. There exists a GNPT $A'$
  running on $\Sigma'$-labeled trees with $2^{\mathcal{O}(k n^2
    \cdot \log k \cdot \log b^2)}$ states, index $nk$, and
  $b$-counting constraints such that $A'$ accepts a tree iff $A$
  accepts its projection on $\Sigma$.
\end{thm}

\proof Let $A = \genA$ with $\FF = \{ \FF_1,\ldots, \FF_k \}$. The
automaton $A'$ is the intersection of three automata $A_1$, $A_2$, and
$A_3$. The automaton $A_1$ is a \safety GNPT, and it accepts a tree
$\pairs{T, (V, \strat, \prom, \ann)}$ iff \strat is a strategy on $V$
and \prom is a promise on $\strat$. It is similar to the corresponding
automaton in \cite{KSV02}, but additionally has to take into account
the capability of 2GAPTs to travel upwards. The state set of $A_1$ is
$Q_1:=2^{(Q \times Q) \cup Q}$. Let $P \in Q_1$.  Intuitively,
\begin{enumerate}[(1)]

\item[(a)] pairs $(q,q') \in P$ represent obligations for \prom in the
  sense that if a node $x$ of an input tree receives state $P$ in a
  run of $A$, then $(q,q')$ is obliged to be in $\prom(x)$;

\item[(b)] states $q \in P$ are used to memorize $\head(\strat(y))$ of
  the predecessor $y$ of $x$.

\end{enumerate}
This behaviour is easily implemented via $A_1$'s transition relation.
Using $\mathsf{false}$ in the transition function of $A_1$ and thus
ensuring that the automaton blocks when encountering an undesirable
situation, it is easy to enforce Conditions~(2) to~(3) of strategies,
and Condition~(3) of promises. The initial state of $A_1$ is
$\{(q_0,q_0)\}$, which together with Condition~(3) of promises
enforces Condition~(1) of strategies. It thus remains to treat
Conditions~(1) and~(2) of promises. This is again straightforward
using the transition function. For example, if $(q,\langle n
\rangle,q') \in \strat(x)$, then we can use the conjunct $\langle
(q,q'), ({>}, n)\rangle$ in the transition. Details of the definition
of $A_1$ are left to the reader. Clearly, the automaton $A_1$ has
$2^{\mathcal{O}(n^2)}$ states and counting bound $b$.

The remaining automata $A_2$ and $A_3$ do not rely on the gradedness
of GNPTs. The automaton $A_2$ is both a \safety and \forrall GNPT. It
accepts a tree $\pairs{T, (V, \strat, \prom, \ann)}$ iff $\ann$ is an
annotation.  More precisely, $A_2$ checks that all conditions of
annotations hold for each node $x$ of the input tree.  The first two
conditions are checked locally by analyzing the labels $\strat(x)$ and
$\ann(x)$.  The last two conditions require to analyze $\prom(x)$,
$\strat(y)$, and $\ann(y)$, where $y$ is the parent of $x$. To access
$\strat(y) \subseteq {Q \times D_b^- \times Q}$ and $\ann(y) \subseteq
{Q \times \{1,\ldots, k\}\times Q} $ while processing $x$, $A_2$ must
memorize these two sets in its states.  Regarding $\strat(y)$, it
suffices to memorize the representation from $L_\strat$.  The number
of such representations is $(4b^2)^{n^2}$, which is bounded by
$2^{\mathcal{O}(n^2 \cdot \log b^2)}$.  There are $2^{k n^2}$
different annotations, and thus the overall number of states of $A_2$
is bounded by $2^{\mathcal{O}(kn^2 \cdot \log b^2)}$.

The automaton $A_3$ is a \forrall GNPT, and it accepts a tree
$\pairs{T, (V, \strat, \prom, \ann)}$ iff $\ann$ is accepting.
By Lemma~\ref{annotaion must be accepting}, it thus follows that $A'$
accepts $\pairs{T, (V, \strat, \prom, \ann)}$ iff $A$ accepts
$\pairs{T, V}$.
The automaton $A_3$ extends the automaton considered in \cite{Var98}
by taking into account promise trees and graded moves in strategies. %

We construct $A_3$ in several steps. We first define a
nondeterministic parity word automaton (NPW) $U$ over $\Sigma'$.
An input word to $U$ corresponds to a path in an input tree to $A'$.
We build $U$ such that it accepts an input word/path if this path
gives rise to a downward trace that violates the acceptance condition
$\FF$ of $A$.  An NPW is a tuple $\pairs{\Sigma, S, M, s_0, \FF}$,
where $\Sigma$ is the input alphabet, $S$ is the set of states, $M:S
\rightarrow 2^S$ is the transition function, $s_0 \in S$ is the
initial state, and $\FF = \{ \FF_1, \FF_2\dots,\FF_k \}$ is a
parity acceptance condition.  Given a word $w=a_0 a_1 \ldots \in
\Sigma^\omega$, a run $r=q_0 q_1 \cdots$ of $U$ on $w$ is such that
$q_0=s_0$ and $q_{i+1} \in M(q_i, a_i)$ for all $i\geq 0$.

We define $U=\pairs{\Sigma', S, M, s_0, \FF'}$ such that $S = (Q
\times Q \times \{1,\ldots, k\}) \cup \{q_{acc}\}$. Intuitively, a run
of $U$ describes a downward trace induced by \strat, \prom, and \ann
on the input path. Suppose that $x$ is the $i$-th node in an input
path to $U$, $r$ is a run of $U$ on that path, and the $i$-th state in
$r$ is $\pairs{q, q_{prev}, j}$. This means that $r$ describes a trace
in which the state of $A$ on the node $x$ is $q$, while the previous
state at the parent $y$ of $x$ was $q_{prev}$.  Thus, $A$ has executed
a transition $(\pairs{b},q)$ or $([b],q)$ to reach state $q$ at $x$.
For reaching the state $q_{prev}$ at $y$, $A$ may or may not have
performed a detour at $y$ as described by $\ann$.  The $j$ in
$\pairs{q, q_{prev}, j}$ is the minimum index of $q$ and any state
encountered on this detour (if any).

We now define the transition function $M$ formally. To this end, let $\pairs{q,
  q_{prev}, j} \in S$ and let $\sigma=(V(x),\strat(x),\prom(x),\ann(x))$.
To define $M(\pairs{q, q_{prev}, j},\sigma)$, we distinguish between
three cases:
\begin{enumerate}[(1)]

\item if $(q_{prev},q) \not \in \prom(x)$, then $M(\pairs{q, q_{prev},
    j}, \sigma)= \emptyset$;

\item otherwise and if $H= \{c: (q,c,q) \in \ann(x)\}$ is non-empty
  and some member of $H$ has an odd index, set $M(\pairs{q, q_{prev},
    j}, \sigma)= \{q_{acc}\}$;

\item if neither (1) nor (2) apply, then we put $\pairs{q', q_{prev}', j'}
  \in M(\pairs{q, q_{prev}, j},\sigma)$ iff

\begin{enumerate}[$\bullet$]
\item $(q,c,q')\in \strat(x)$, with $c \in \ppairs{b} \cup [[b]]$,
$q_{prev}'=q$ and $j'=\mathsf{index}(q')$; or

\item
$(q,d,q_{prev}')\in \ann(x)$ for some $d$,
$(q_{prev}',c,q')\in \strat(x)$ for some $c \in \ppairs{b} \cup [[b]]$,
and $j'=\mathsf{min}(d,\mathsf{index}(q'))$.

\end{enumerate}
\end{enumerate}
In addition, $M(q_{acc},\sigma)=\{q_{acc}\}$, for all $\sigma \in
\Sigma'$. For (1), note that if $(q_{prev},q) \not \in \prom(x)$, then
\prom does not permit downwards traces in which $A$ switches from
$q_{prev}$ to $q$ when moving from the parent of $x$ to $x$. Thus, the
current run of $U$ does not correspond to a downward trace, and $U$
does not accept. The purpose of (2) is to check for traces that ``get
caught'' at a node.

The initial state $s_0$ of $U$ is defined as $\pairs{q_0, q_0,
  \ell}$, where $\ell$ is such that $q_0 \in \FF_\ell$.  Note that the
choice of the second element is arbitrary, as the local promise at the
root of the input tree is irrelevant.
Finally, the parity condition is $\FF' = \{ \FF_1', \FF_2', \ldots, \FF'_{k+1}
\}$, where $\FF'_1=\emptyset$, $\FF'_2= Q \times Q \times \{1\} \cup
\{q_{acc}\}$ and for each $\ell$ with $2 < \ell \leq k+1$, we have $\FF'_\ell =
Q \times Q \times \{\ell-1\}$. Thus, $U$ accepts a word if this word
corresponds to a path of the input tree on which there is a non-accepting
trace.

In order to get $A_3$, we co-determinize the NPW $U$ and expand it to
a tree automaton, i.e., a \forrall GNPT on $\Sigma'$. That is, we
first construct a deterministic parity word automaton $\widetilde{U}$
that complements $U$, and then replace a transition
$\widetilde{M}(q,\sigma) = q'$ in $\widetilde{U}$ by a transition
$$
M_{t}(q , \sigma) = \{\pairs{(\neg \theta_{q'}),(\leq,0)}\}
$$
in $A_3$ where the states of $\widetilde{U}$ are encoded by some set
$Y$ of variables and for every state $q'$, the formula $\theta_{q'}
\in B(Y)$ holds only in the subset of $Y$ that encodes $q'$. By
\cite{Saf89,Tho97}, the automaton $\widetilde{U}$ has $(nk)^{nk} \leq 2^{nk \cdot \log nk}$
states and index $nk$, thus so does $A_3$. 

By Lemma~\ref{intersectiontwo}, we can intersect the two \safety
automata $A_1$ and $A_2$ obtaining a \safety automaton with
$2^{\mathcal{O}(kn^2 \cdot \log b^2)}$ states and counting bound $b$.
Moreover, by Lemma~\ref{intersection}, the obtained \safety automaton
can be intersected with the \forrall automaton $A_3$ yielding the
desired GNPT $A'$ with $2^{\mathcal{O}(k n^2 \cdot \log k \cdot \log
  b^2)}$ states, counting bound $b$, and index $nk$.
\endproof

\subsection{Emptiness of GNPTs}

By extending results of \cite{KV98,KVW00,KSV02}, we provide an
algorithm for deciding emptiness of GNPTs. The general idea is to
translate GNPTs into alternating (non-graded) parity automata on
words, and then to use an existing algorithm from \cite{KV98} for
deciding emptiness of the latter.

A \emph{singleton-alphabet GNPT on full $\omega$-trees
  ($\omega$-1GNPT)} is a GNPT that uses a singleton alphabet $\{a\}$
and admits only a single input tree $\pairs{T_\omega,V}$, where
$T_\omega$ is the full $\omega$-tree $\Nat^+$ and $V$ labels every
node with the only symbol $a$.  Our first aim is to show that
every GNPT can be converted into an $\omega$-1GNPT such that
(non)emptiness is preserved. We first convert to a 1GNPT, which is a
single-alphabet GNPT.
\begin{lem}
  Let $A=\pairs{\Sigma,b,Q,\delta,q_0,\FF}$ be a GNPT. Then there is a
  1GNPT $A'=\pairs{\{a\},b,Q',\delta',q'_0,\FF'}$ with
  $L(A)=\emptyset$ iff $L(A')=\emptyset$ and $|Q'| \leq |Q| \times
  |\Sigma|+1$.%
\end{lem}
\begin{proof}
  Let $Q \subseteq 2^Y$.  We may assume w.l.o.g.\ that $\Sigma
  \subseteq 2^Z$ for some set $Z$ with $Z \cap Y = \emptyset$.
  Now define the components of $A'$ as follows:
  \begin{enumerate}[$\bullet$]

  \item $Q'=\{\{ s \}\} \cup \{ q \cup \sigma, \mid q \in Q \wedge \sigma \in \Sigma \}
    \subseteq 2^{Y'}$, where $Y'=Y \uplus Z \uplus \{ s \}$;

  \item $q'_0=\{s\}$;

  \item $\delta'(\{s\},a)=\{ \pairs{\mathsf{true},({\leq},1)},
    \pairs{\bigwedge_{y \in q_0} y \wedge \bigwedge_{y \in Y \setminus
        q_0} \neg y, ({>},0)}, \pairs{s,({\leq},0)} \}$;

  \item $\delta'(q,a)=\delta(q \cap Y,q \cap Z) \cup \{
    \pairs{s,({\leq},0)} \}$ for all $q \in Q$ with $q \neq \{ s \}$; 

  \item $\FF' = \{ \FF'_1, \dots, \FF'_k \}$ with $\FF'_{i} = \{ q \in Q'
    \mid q \cap Q \in \FF_{i} \}$ if $\FF = \{ \FF_{1},\dots,\FF_{k} \}$.

  \end{enumerate}
  It is easy to see that $A$ accepts $\langle T,V \rangle$ iff $A'$
  accepts $\pairs{T',V'}$, where $T'$ is obtained from $T$ by adding
  an additional root, and $V'$ assigns the label $a$ to every node in
  $T'$. Intuitively, the additional root enables $A'$ to ``guess'' a
  label at the root of the original tree. Then, the label will be guessed
  iteratively.
\end{proof}
In the next step, we translate to $\omega$-1GNPTs.
\begin{lem}
  Let $A=\pairs{\{a\},b,Q,\delta,q_0,\FF}$ be a 1GNPT. Then there
  exists an $\omega$-1GNPT \break ${A'=\pairs{\{a\},b,Q',\delta',q_0,\FF'}}$
  such that $\mathcal{L}(A)=\emptyset$ iff $\mathcal{L}(A')=\emptyset$ and
  $|Q'|=|Q|+1$.%
\end{lem}
\begin{proof}
  Define the components of $A'$ as follows:
  \begin{enumerate}[$\bullet$]

  \item $Q'=Q \cup \{ \{ \bot \} \} \subseteq 2^{Y'}$, where $Y'=Y
    \uplus \{ \bot \}$;

  \item if
    $\delta(q,a)=\{\pairs{\theta_1,\xi_1},\dots,\pairs{\theta_k,\xi_k}\}$,
    set $\delta'(q,a)=\{\pairs{\theta_1 \wedge \neg
      \bot,\xi_1},\dots,\pairs{\theta_k \wedge \neg
      \bot,\xi_k}\}$, for all $q \in Q$ with $\bot \notin q$;

  \item $\delta'(q,a)=\{ \pairs{\neg \bot, ({\leq}, 0)} \}$ for all
    $q \in Q$ with $\bot \in q$.

  \item $\FF' = \{ \FF'_1, \dots, \FF'_k \}$ with
$ \FF'_{1} = \FF_{1}$, and
$ \FF'_{i} = \FF_i \cup \{ q \in Q' \mid \bot \in q \} $, for $2 \leq i \leq
k$,
    if $\FF = \{ \FF_{1},\dots,\FF_{k} \}$.
  \end{enumerate}
  It is easy to see that $\mathcal{L}(A) \neq \emptyset $ iff $A'$ accepts
  $\pairs{T_\omega,V}$.  Accepting runs can be translated back and forth. When
  going from runs of $A$ to runs of $A'$, this involves of the children
  of each node with nodes labeled $\{\bot\}$.
\end{proof}
We are now ready to translate GNPTs to alternating word automata.  A
\emph{single-alphabet alternating parity word automaton (1APW)} is a
tuple $A=\pairs{\{a\},Q,\delta,q_0,\FF}$, where $\{a\}$ is the
alphabet, $Q$, $q_0$, and $\FF$ are as in FEAs, and $\delta:Q \times
\{a\} \rightarrow B^+(Q)$. There is only a single possible input to a
1APW, namely the infinite word $aaa\cdots$. Intuitively, if $A$ is in
state $q$ on the $i$-th position of this word and $\delta(q,a)=q' \vee
(q \wedge q'')$, then $A$ can send to position $i+1$
either a copy of itself in state $q'$ or one copy in state $q$ and one in state
$q''$. The input word is accepted iff there is an accepting run of
$A$, where a \emph{run} is a $Q$-labeled tree $\pairs{T_r,r}$ such
that
\begin{enumerate}[$\bullet$]

\item $r(\mathsf{root}(T_r)) = q_0$;

\item for all $y \in T_r$ with $r(y) = q$ and $\delta(q , a) =
  \theta$, there is a (possibly empty) set $S \subseteq Q$ such that
  $S$ satisfies $\theta$ and for all $q' \in S$, there is $j \in \Nat$
  such that $y \cdot j \in T_r$ and $r(y \cdot j) = q'$.

\end{enumerate}
As for FEAs, a run $\pairs{T_r,r}$ is \emph{accepting} if all its
infinite paths satisfy the acceptance condition.

\medskip

For an $\omega$-1GNPT $A=\pairs{\{a\},b,Q,\delta,q_0,\FF}$, $q \in Q$,
and $P \subseteq Q$, the function $\mathsf{is\_mother}_A(q,P)$ returns
$\mathsf{true}$ if there is an infinite word $t \in P^\omega$ that
satisfies the counting constraint $\delta(q,a)$, and $\mathsf{false}$
otherwise.
\begin{lem}
\label{gnptto1apw}
For every $\omega$-1GNPT $A=\pairs{\{a\},b,Q,\delta,q_0,\FF}$, the 1APW
$A'=\pairs{\{a\},Q,\delta',q_0,\FF}$ is such that $\mathcal{L}(A) =\emptyset$ iff
$\mathcal{L}(A')=\emptyset$, where for all $q \in Q$,
  $$
  \delta'(q,a)= \bigvee_{P \subseteq Q \text{ s.t. } \mathsf{is\_mother}_A(q,P)} \ \
  \bigwedge_{q \in P} q.
  $$
\end{lem}
\begin{proof}
  (sketch) First assume that $\pairs{T_\omega,V} \in \mathcal{L}(A)$. Then there exists an
  accepting run $\pairs{T_\omega,r}$ of $A$ on $\pairs{T_\omega,V}$. It is not
  difficult to verify that $\pairs{T_\omega,r}$ is also an accepting
  run of $A'$. Conversely, assume that $a^\omega \in \mathcal{L}(A')$. Then
  there is an accepting run $\pairs{T_r,r}$ of $A'$.
  We define an accepting run $\pairs{T_\omega,r'}$ of $A$ on
  $\pairs{T_\omega,V}$ by inductively defining $r'$. Along with $r'$, we define
  a mapping $\tau:T_\omega \rightarrow T_r$ such that
  $r'(x)=r(\tau(x))$ for all $x \in T_\omega$. To start, set
  $r'(\mathsf{root}(T_\omega))=q_0$ and
  $\tau(\mathsf{root}(T_\omega))=\mathsf{root}(T_r)$. For the
  induction step, let $x \in T_\omega$ such that $r'(y)$ is not yet
  defined for the successors $y$ of $x$. Since $\pairs{T_r,r}$ is a
  run of $A'$ and by definition of $\delta'$, there is a $P \subseteq
  Q$ such that (i)~$\mathsf{is\_mother}_A(r(\tau(x)),P)$ and (ii)~for
  all $q \in P$, there is a successor $y$ of $\tau(x)$ in $T_r$ with
  $r(y) =q$. By (i), there is a word $t=q_1q_2\cdots \in P^\omega$ that
  satisfies the counting constraint
  $\delta(r(\tau(x)),a)=\delta(r'(x),a)$. For all $i \geq 1$, define
  $r'(x \cdot i) = q_i$ and set $\tau(x \cdot i)$ to some successor
  $y$ of $\tau(x)$ in $T_r$ such that $r(y)=q_i$ (which exists by
  (ii)).  It is not hard to check that $\pairs{T_\omega,r'}$ is indeed
  an accepting run of $A$ on $\pairs{T_\omega,V}$.
\end{proof}
For a 1APW $A=\pairs{\{a\},Q,\delta,q_0,\FF}$, $q \in Q$, and $P
\subseteq Q$, the function $\mathsf{is\_mother}_A(q,P)$ returns
$\mathsf{true}$ if $P$ satisfies the Boolean formula $\delta(q,a)$,
and $\mathsf{false}$ otherwise.

Since the transition function of the automaton $A'$ from
Lemma~\ref{gnptto1apw} is of size exponential in the number of states
of the $\omega$-1GNPT $A$, we should not compute $A'$ explicitly.
Indeed, this is not necessary since all we need from $A'$ is access to
$\FF$ and $\mathsf{is\_mother}_{A'}$ and, as stated in the next lemma,
$\mathsf{is\_mother}_{A'}$ coincides with $\mathsf{is\_mother}_A$.
The lemma is an immediate consequence of the definition of the 1APW in
Lemma~\ref{gnptto1apw}.
\begin{lem}
\label{allmothersarethesame}
Let $A$ and $A'$ be as in Lemma~\ref{gnptto1apw}, with state set $Q$.
Then $\mathsf{is\_mother}_A = \mathsf{is\_mother}_{A'}$.
\end{lem}
To decide the emptiness of 1APWs, we use the algorithm from
\cite{KV98}. It is a recursive procedure that accesses the transition
function of the 1APW only via $\mathsf{is\_mother}$. If started on a
1APW with $n$ states and index $k$, it makes at most $2^{\mathcal{O}(k
  \log n)}$ calls to $\mathsf{is\_mother}$ and performs at most
$2^{\mathcal{O}(k
  \log n)}$ additional
steps. %

To analyze its runtime requirements, we first determine the complexity
of computing $\mathsf{is\_mother}$.\footnote{We remark that the
  analogous Lemma~1 of \cite{KSV02} is flawed because it considers
  only trees of finite outdegree.}
\begin{lem}
\label{lem:compismother}
  Let $A=\pairs{\{a\},b,Q,\delta,q_0,\FF}$ be an $\omega$-1GNPT with
  $n$ states and counting bound $b$. Then $\mathsf{is\_mother}_A$ can
  be computed in time $b^{\mathcal{O}(\log n)}$.
\end{lem}
\begin{proof}
  Assume that we want to check whether $\mathsf{is\_mother}_A(q,P)$,
  for some $q \in Q$ and $P \subseteq Q$. Let
  $\theta_1,\dots,\theta_k$ be all formulas occurring in
  $C:=\delta(q,a)$. We construct a deterministic B\"uchi automaton
  $A'=\pairs{\Sigma',Q',q'_0,\delta',F'}$ on infinite words that
  accepts precisely those words $t \in P^\omega$ that satisfy $C$:
  \begin{enumerate}[$\bullet$]

  \item $\Sigma'=P$;

  \item $Q'=\{0,\dots,b\}^k$;

  \item $q'_0 = \{0\}^k$;

  \item $\delta'((i_1,\dots,i_k),p)$ is the vector $(j_1,\dots,j_k)$,
    where for all $h \in \{1,\dots,k\}$, we have \break ${j_h=\min \{b,
    i_h+1\}}$ if $p \in \Sigma'$ satisfies $\theta_h$, and $j_h =i_h$
    otherwise;

  \item $F'$ consists of those tuples $(i_1,\dots,i_k)$ such that
    for all $h \in \{1,\dots,k\}$,
    \begin{enumerate}[(1)]

    \item there is no $\pairs{\theta_h,(\leq, r)} \in C$ with
      $r < i_h$;

    \item for all  $\pairs{\theta_h,(>, r)} \in C$, we have
      $i_h \geq r$.

    \end{enumerate}

  \end{enumerate}
  By definition of GNPTs, the cardinality of $C$ is bounded by $\log
  n$. Thus, $A'$ has $b^{\log n}$ states. It remains to note that
  the emptiness problem for deterministic B\"uchi word automata
  (is \nlsp-complete~\cite{VW94}
  and) can be solved in linear time~\cite{Var07}.
\end{proof}
Now for the runtime of the algorithm.  Let $A$ be a GNPT with $n$
states, %
counting bound $b$, and index $k$. To decide emptiness of $A$, we
convert $A$ into an $\omega$-1GNPT $A'$ with $n+1$ states, counting
bound $b$, and index $k$, and then into a 1APW $A''$ with $n+1$ states
and index $k$.  By Lemma~\ref{lem:compismother}, we obtain the
following result.
\begin{thm}
\label{thm:gnptempty}
  Let $A=\pairs{\Sigma,b,Q,\delta,q_0,\FF}$ be a GNPT with $|Q|=n$,
  and index $k$. Then emptiness of $A$ can be decided
  in time $(b+2)^{\mathcal{O}(k \cdot \log n )}$.
\end{thm}

\subsection{Wrapping Up}

Finally, we are ready to prove Theorem~\ref{emptiness}, which we restate here for convenience.
\\[2mm]
{\bf Theorem~\ref{emptiness}.} \emph{
The emptiness problem for a 2GAPT $A = \genA$ with %
$n$ states and index $k$ can be solved in time $(b+2)^{\mathcal{O}(n^2 \cdot
k^2 \cdot \log k \cdot \log b^2)}$.}
\\[-2mm]
\begin{proof}
  By Theorem~\ref{translation}, we can convert $A$ into a GNPT $A'$
  with $2^{\mathcal{O}(kn^2 \cdot \log k \cdot \log b^2)}$ states, index
  $nk$, and counting bound $b$.
  Thus, Theorem~\ref{thm:gnptempty} yields the desired
  result.
\end{proof}
A matching {\sc ExpTime} lower bound is inherited from nongraded,
one-way alternating tree automata.

\section{Conclusion}
\label{sect:concl}

We have studied the complexity of $\mu$-calculi enriched with inverse
programs, graded modalities, and nominals. Our analysis has resulted
in a rather complete picture of the complexity of such logics. In
particular, we have shown that only the fully enriched $\mu$-calculus
is undecidable, whereas all its fragments obtained by dropping at
least one of the enriching features inherit the attractive
computational behavior of the original, non-enriched $\mu$-calculus.

From the perspective of the description logic OWL, the picture is as
follows. Undecidability of the fully enriched $\mu$-calculus means
that OWL extended with fixpoints is undecidable. The decidable
$\mu$-calculi identified in this paper give rise to natural fragments
of OWL that remain decidable when enriched with fixpoints. Orthogonal
to the investigations carried out in this paper, it would be
interesting to understand whether there are any second-order features
that can be added to OWL without losing decidability. In particular,
decidability of OWL extended with transitive closure is still an open
problem.

\subsection*{Acknowledgements} We are grateful to Orna Kupferman and
Ulrike Sattler for helpful discussions of \cite{SV01,KSV02}.

\bibliographystyle{alpha}

\end{document}